\documentclass[11pt]{article}

\usepackage[margin=1in]{geometry}
\usepackage{amsmath,amssymb,amsthm}
\usepackage{braket}
\IfFormatAtLeastTF{2026-06-01}
  {}
  {\usepackage{aliascnt}}
\usepackage{algorithm}
\usepackage{algpseudocode}
\usepackage{microtype}
\usepackage[dvipsnames]{xcolor}
\usepackage[pdftex,bookmarks, plainpages=false, pdfpagelabels=true]{hyperref}
\hypersetup{
    bookmarksnumbered=true, 
    unicode=false, 
    pdfstartview={FitH}, 
    pdftitle={Near-Optimal Separations of Certificate Complexity from Randomized and Quantum Query Complexity}, 
    pdfauthor={Andris Ambainis, Jānis Iraids, Martins Kokainis}, 
    pdfsubject={}, 
    pdfcreator={}, 
    pdfproducer={}, 
    pdfkeywords={}, 
    pdfnewwindow=true, 
    colorlinks=true, 
    urlcolor=WildStrawberry, 
    linkcolor=NavyBlue!75, 
    citecolor=ForestGreen, 
    filecolor=BrickRed 
}
\usepackage[nameinlink,noabbrev]{cleveref}
\crefname{appendix}{appendix}{appendices}
\Crefname{appendix}{Appendix}{Appendices}
\crefname{subappendix}{appendix}{appendices}
\Crefname{subappendix}{Appendix}{Appendices}
\makeatletter
\renewcommand*{\theHALG@line}{\thealgorithm.\arabic{ALG@line}}
\makeatother

\newtheorem{theorem}{Theorem}[section]
\IfFormatAtLeastTF{2026-06-01}{
  \newtheorem{lemma}[theorem]{Lemma}
}{
  \newaliascnt{lemma}{theorem}
  \newtheorem{lemma}[lemma]{Lemma}
  \aliascntresetthe{lemma}
}

\theoremstyle{definition}
\IfFormatAtLeastTF{2026-06-01}{
  \newtheorem{definition}[theorem]{Definition}
}{
  \newaliascnt{definition}{theorem}
  \newtheorem{definition}[definition]{Definition}
  \aliascntresetthe{definition}
}
\theoremstyle{plain}

\title{Near-Optimal Separations of Certificate Complexity from Randomized and Quantum Query Complexity}

\usepackage[affil-it]{authblk}

\author[1]{Andris Ambainis}
\author[1]{J\={a}nis Iraids}
\author[1]{Martins Kokainis}

\affil[1]{Center for Quantum Computer Science, Faculty of Science and Technology, University of Latvia}
\date{}

\begin{document}

\maketitle

\begin{abstract}
We study how large the certificate complexity $\operatorname{C}(f)$ of a total Boolean function can be relative to its randomized and quantum query complexities. We construct a total Boolean function $f$ whose randomized query complexity with one-sided error satisfies $\operatorname{R}_1(f) = \Theta(\sqrt{\operatorname{C}(f)})$. This separation is optimal even when two-sided error is allowed.

From this construction, we obtain another total Boolean function $F$ with bounded-error quantum query complexity $\operatorname{Q}(F) = \widetilde O(\operatorname{C}(F)^{1/4})$, attaining the general quartic bound up to polylogarithmic factors. For the same function, both exact and zero-error quantum query complexities are $\widetilde O(\sqrt{\operatorname{C}(F)})$. We also prove matching lower bounds for these two measures on $F$.
\end{abstract}

\section{Introduction}
Certificate complexity $\operatorname{C}(f)$ 
is one of the most fundamental combinatorial complexity measures for Boolean functions (along with block sensitivity and polynomial
degree). $\operatorname{C}(f)$ is tightly connected to query complexity of $f$ in deterministic, randomized and quantum models of computation.
Let $\operatorname{D}(f), \operatorname{R}(f)$ and $\operatorname{Q}(f)$ denote these three complexity measures (with bounded error in the latter two cases).
For total Boolean functions $f$, we have 
\begin{itemize}
    \item 
    $\operatorname{C}(f) \leq \operatorname{D}(f) \leq \operatorname{C}(f)^2$ \cite{Tardos1989},
    \item
    $\operatorname{R}(f) = \Omega(\sqrt{\operatorname{C}(f)})$ and $\operatorname{R}(f) \leq \operatorname{C}(f)^2$ \cite{Nisan1991}, and
    \item
    $\operatorname{Q}(f) = \Omega(\sqrt[4]{\operatorname{C}(f)})$ and $\operatorname{Q}(f) \leq \operatorname{C}(f)^2$ \cite{Beals2001}.
\end{itemize}
In particular, this implies that deterministic, randomized and quantum complexities are all polynomially related. 

Several of the best-known relations among $\operatorname{D}(f)$, $\operatorname{R}(f)$, and $\operatorname{Q}(f)$ have used certificate complexity as an intermediate measure. For example, Nisan’s bound $\operatorname{D}(f)=O(\operatorname{R}(f)^3)$ \cite{Nisan1991}, still the best general relation between randomized and deterministic query complexity, proceeds through $\operatorname{C}(f)$. The earlier bound $\operatorname{D}(f)=O(\operatorname{Q}(f)^6)$ of Beals et al. \cite{Beals2001} likewise uses certificate complexity, although the more recent bound $\operatorname{D}(f)=O(\operatorname{Q}(f)^4)$ \cite{Aaronson2021} does not.

Given the fundamental nature of $\operatorname{C}(f)$ and its close relationship to query complexity measures, it is interesting to settle the questions about maximum possible gap between $\operatorname{C}(f)$ and various randomized and quantum query complexities. The research in the last decade \cite{Aaronson2021,cornelissen} has settled many of the other gaps between query and combinatorial complexity measures for total Boolean functions but these questions have remained open.

For the $\operatorname{Q}$-vs-$\operatorname{C}$ question, the best previously known result was 
\[
\operatorname{C}(f)=\Theta(\operatorname{Q}(f)^2),
\]
for the OR function. On the other hand, the general bound $\operatorname{C}(f)=O(\operatorname{Q}(f)^4)$ \cite{Aaronson2021} leaves open the possibility of a quartic separation. Thus the optimal exponent relating $\operatorname{C}(f)$ to $\operatorname{Q}(f)$ was known only to lie between 2 and 4.

In this paper, we close the gap (up to polylogarithmic factors), by constructing a function $F$ for which $\operatorname{Q}(F) = \widetilde{O}(\sqrt[4]{\operatorname{C}(F)})$.

The same function also substantially improves the best known separations between certificate complexity and exact and zero-error quantum query complexities. Let $\operatorname{Q}_{\mathrm E}(f)$ and $\operatorname{Q}_0(f)$ denote these two measures. Here, the biggest known separation was $\operatorname{C}(f)=\Omega(\operatorname{Q}_{\mathrm E}(f)^{1.15...})$ \cite{ambainis2013}. For our function, both 
$\operatorname{Q}_{\mathrm E}(F)$ and $\operatorname{Q}_0(F)$ are $\widetilde{O}(\sqrt{\operatorname{C}(F)})$, significantly improving over the best previously known separation. 

Unlike the bounded error result, it is not known whether these results are optimal: the best upper bounds are $\operatorname{C}(f)=O(\operatorname{Q}_{\mathrm E}(f)^3)$ and $\operatorname{C}(f)=O(\operatorname{Q}_0(f)^3)$ which follow from $\operatorname{D}(f)=O(\operatorname{Q}_{\mathrm E}(f)^3)$ \cite{midrijanis2004} and $\operatorname{D}(f)=O(\operatorname{Q}_0(f)^3)$ \cite[Theorem~6]{midrijanis2005}.

For the $\operatorname{R}$-vs-$\operatorname{C}$ question, we obtain a logarithm-free separation: $\operatorname{R}(f)=O(\sqrt{\operatorname{C}(f)})$. Along with the simultaneous work by Ben-David and Kothari \cite{BenDavidKothari2026}, this is the first example of an asymptotic separation between these two complexity measures.
Since $\operatorname{C}(f)=O(\operatorname{R}(f)^2)$ for any total $f$ \cite{Nisan1991}, this result is optimal.

Our construction builds on the work of Pabbaraju \cite{pabbaraju2026}, who constructed a family of total Boolean functions satisfying $\operatorname{C}(f) = \widetilde{\Omega}(\widetilde{\operatorname{deg}}(f)^4)$. Here, $\widetilde{\operatorname{deg}}(f)$ is the least degree of a real polynomial that approximates $f$ with error at most $1 / 3$ on every input. 
Given the connection between quantum algorithms and polynomials, it is plausible to expect a quantum algorithm of a similar complexity. However, the approximate degree can be asymptotically smaller than quantum query complexity. Thus the approximate-degree separation does not by itself provide the quantum upper bound required here.

We construct randomized and quantum algorithms for computing modifications of Pabbaraju's function.

In more detail, \cite{pabbaraju2026} uses a partial function whose domain is specified by a disjunctive normal form (DNF) with mutually incompatible conjunctions of input literals. A satisfied conjunction identifies a block and a set of positions within it; the partial function returns the OR of additional input bits at those positions. The cheat-sheet construction \cite{AaronsonBenDavidKothari2016} totalizes this partial function using an array of certificate descriptions; at most one cell is valid, and the total function is one if and only if such a cell exists. \cite{pabbaraju2026} obtains the approximate degree upper bound by summing polynomials that approximate the individual cell tests.

We also use the cheat-sheet totalization of the partial function and retain the way the partial function is constructed from the underlying DNF, while modifying how the DNF is defined to obtain a slightly different total Boolean function. The underlying DNF achieves the quadratic separation between $\operatorname{R}(f)$ and $\operatorname{C}(f)$.

On positive inputs, our quantum algorithm first computes the values determining the cell address and then verifies the addressed cell. Computing each value requires locating the block identified by the satisfied conjunction. We locate this block by finding a source in a directed complete graph defined by comparisons between blocks. The algorithm adapts the variable-time search construction of \cite{AmbainisKokainisVihrovs}; its query bound follows from estimates for the successive candidate sets and the query counts of their membership tests.

\paragraph*{Concurrent work.}
Concurrently and independently, \cite{BenDavidKothari2026} constructs a different total Boolean function with certificate complexity $\Omega(n^2)$, randomized query complexity $O(n\ln n)$, and bounded-error quantum query complexity $\widetilde O(\sqrt{n})$. Its input includes lists of candidate labels for pairs of groups, together with claims about the candidates' positions in those lists. Checking these claims permits comparisons using $O(\ln n)$ queries and an application of a known algorithm for finding a sink.

One can also obtain $\Theta(n\ln n)$ bounds for its exact and zero-error quantum query complexities as follows. For the exact upper bound, the candidate sampling in \cite[Lemma~3]{BenDavidKothari2026} is modified to choose one list independently and uniformly per group, retaining the first entry only if it names a candidate from that group. On positive inputs, the success probability is at least $1/2$ and can be computed from an accepted description, so the weighting argument in \Cref{sec:exact-query} applies. For the lower bound, a restriction fixes one candidate first in every list involving its group and varies only its description. The restricted function is one exactly when all $\Theta(n\ln n)$ bits of that description are zero, giving the zero-error lower bound by the argument of \Cref{lem:zero-error-lower}.

The construction in \cite{BenDavidKothari2026} suggests an alternative totalization of the DNF $f(x)$ on $n$ blocks. One can define a total function $F'_n$ whose input contains $x$ and descriptions $z_{i,p}$ indexed by a block and a position. A description is accepted when it verifies that $x$ satisfies a term with block index $i$ and that $p$ is the least position containing a one in that block. The function is one if and only if some description is accepted. Its quantum algorithm runs \textsc{FindSource}, finds the least nonzero position in the returned block, and checks the indexed description. This gives
\[
 \operatorname{C}(F'_n) = \Theta(n^2), \qquad
 \operatorname{Q}(F'_n) = O(\sqrt{n}\ln^2 n).
\]
The variant has $\Theta(n^3\ln n)$ input bits and randomized query complexity $O(n\ln n)$; its exact and zero-error quantum query complexities are both $\Theta(n\ln n)$. The construction and bounds are given in \Cref{app:alternative-totalization}.

\section{Definitions and results}\label{sec:results}

\subsection{Notation and complexity measures}

For a positive integer $m$, write $[m] = \{1,\ldots,m\}$, and put $[0] = \emptyset$. The notation $\widetilde O$, $\widetilde\Omega$, and $\widetilde\Theta$ suppresses polylogarithmic factors.

A partial assignment fixes the values of a subset of the input bits; its size is the number of bits it fixes. For a total Boolean function $g : \{0,1\}^m \to \{0,1\}$ and $b \in \{0,1\}$, a $b$-certificate is a partial assignment such that $g(x) = b$ for every completion $x$. For $x \in g^{-1}(b)$, let $\operatorname{C}_b(g,x)$ be the minimum size of a $b$-certificate consistent with $x$. The certificate complexity is $\operatorname{C}(g) = \max_x \operatorname{C}_{g(x)}(g,x)$.

For $b \in \{0,1\}$, let $\operatorname{UC}_b(g)$ be the least $k$ for which there is a family of $b$-certificates of size at most $k$ such that every input in $g^{-1}(b)$ is consistent with exactly one member. Put $\operatorname{UC}(g) = \max_b \operatorname{UC}_b(g)$ and $\operatorname{UC}_{\min}(g) = \min_b \operatorname{UC}_b(g)$.

Let $\operatorname{D}(g)$ denote deterministic query complexity, and let $\operatorname{R}(g)$ and $\operatorname{Q}(g)$ denote bounded-error randomized and quantum query complexity, respectively. The error is at most $1 / 3$ on every input, and query counts are maximized over inputs and executions. Let $\operatorname{R}_1(g)$ denote randomized query complexity with one-sided error: negative inputs are rejected with certainty, and positive inputs are accepted with probability at least $2 / 3$. Let $\operatorname{R}_0(g)$ denote the minimum $T$ for which a randomized algorithm makes at most $T$ queries on every input $x$, outputs $g(x)$ with probability at least $1 / 2$, and otherwise outputs the inconclusive symbol \texttt{?}. For a partial function, correctness is required only on its domain, while the query bound applies on every input.

Let $\operatorname{Q}_{\mathrm E}(g)$ denote the minimum fixed query count of a circuit computing $g$ with certainty. Define $\operatorname{Q}_0(g)$ by replacing randomized algorithms with quantum algorithms in the definition of $\operatorname{R}_0(g)$.

A disjunctive normal form (DNF) is a disjunction of terms, each a conjunction of variables or their negations. It is \emph{unambiguous} if every input satisfies at most one term.

\subsection{Main results}

The relations in \cite[Section~2.3]{AaronsonBenDavidKothari2016} imply the bound $\operatorname{C}(g) = O(\operatorname{Q}(g)^4)$ for every total Boolean function $g$. We show that the exponent four is attained up to logarithmic factors. Independent proofs of the analogous separation between certificate complexity and approximate degree appear in \cite[Theorem~4]{pabbaraju2026} and \cite[Theorem~1.1]{Balodis}.

\begin{theorem}\label{thm:main}
For all sufficiently large $n$, there is a total Boolean function $F_n$ on $\Theta(n^3 \ln^2 n)$ bits such that
\[
 \operatorname{C}(F_n) \geq n \left(\left\lceil \frac{n}{2}\right\rceil + 1\right) = \Omega(n^2),
 \qquad
 \operatorname{Q}(F_n) = O\left(\sqrt{n} \ln^3 n\right).
\]
\end{theorem}

The two bounds in \Cref{thm:main} follow from \Cref{lem:certificate-lower,lem:query-transfer,lem:query-upper}.

The same function also satisfies $\operatorname{R}_0(F_n) = \widetilde\Theta(\operatorname{Q}(F_n)^4)$. This follows from \Cref{thm:main}, $\operatorname{C}(F_n) \leq \operatorname{R}_0(F_n) \leq \operatorname{D}(F_n)$, and the general bound $\operatorname{D}(g) = O(\operatorname{Q}(g)^4)$ for total functions in \cite{Aaronson2021}.

The underlying DNF gives a quadratic separation between randomized query complexity with one-sided error and certificate complexity.

\begin{theorem}\label{thm:randomized-upper}
For all sufficiently large even $n$, the DNF $f$ constructed in \Cref{sec:construction} satisfies
\[
 \operatorname{C}(f) = \Theta(n^2),
 \qquad
 \operatorname{R}_1(f) = O(n).
\]
\end{theorem}

The algorithm is given in \Cref{sec:randomized-algorithm}; the certificate bound follows from \Cref{lem:undefined-input} and the input length. Since $\operatorname{R}(f) \leq \operatorname{R}_1(f)$, the theorem also gives $\operatorname{R}(f) = O(\sqrt{\operatorname{C}(f)})$. The general bound $\operatorname{C}(g) = O(\operatorname{R}(g)^2)$ in \cite[Section~2.3]{AaronsonBenDavidKothari2016} then gives $\operatorname{R}_1(f) = \Theta(\sqrt{\operatorname{C}(f)})$.

\cite{pabbaraju2026} constructs a DNF exhibiting a quadratic separation of both $\operatorname{C}$ and $\operatorname{UC}$ from $\operatorname{UC}_{\min}$; our DNF retains this separation.

For the same total function $F_n$, we also determine the exact and zero-error quantum query complexities.

\begin{theorem}\label{thm:exact-zero-error}
For all sufficiently large $n$, the functions $F_n$ in \Cref{thm:main} satisfy
\[
 \operatorname{Q}_{\mathrm E}(F_n) = \Theta(n \ln^2 n),
 \qquad
 \operatorname{Q}_0(F_n) = \Theta(n \ln^2 n).
\]
\end{theorem}
The bounds follow from the exact algorithm in \Cref{lem:exact-total-upper}, the restriction of the count fields in \Cref{lem:zero-error-lower}, and $\operatorname{Q}_0(F_n) \leq \operatorname{Q}_{\mathrm E}(F_n)$. Together with \Cref{thm:main}, they give both complexities as $\widetilde O(\sqrt{\operatorname{C}(F_n)})$.

\section{Construction and basic bounds}\label{sec:construction}
We define the function $F_n$ that gives the separation in three steps:
\begin{enumerate}
\item
we first define an unambiguous DNF $f(x)$, $ x\in\{0, 1\}^{n^2}$;
\item
we use $f(x)$ to define a partial Boolean function $H(x, y)$, $x, y \in\{0, 1\}^{n^2}$;
\item 
we apply the cheat-sheet construction to $H(x, y)$ to obtain $F_n$.
\end{enumerate}

For better flow of the text,  the proofs of most lemmas are postponed to \Cref{app:construction}.

\subsection{Unambiguous DNF and partial function \texorpdfstring{$H$}{H}}

Write each $x, y \in \{0,1\}^{n^2}$ as $n$ blocks of length $n$, with $x_i(p)$ and $y_i(p)$ denoting the bits at position $p$ in block $i$.
Let $k = \lfloor n / 2 \rfloor$.
The unambiguous DNF $f(x)$ is defined as
\[
 f(x) = \bigvee_{i \in [n]}\ \bigvee_{\substack{S \subseteq [n]\\ |S| = k}} T_{i, S}, \]
\[ T_{i, S} =  \left(\bigwedge_{p \in S} x_i(p)\right)
 \left(\bigwedge_{p \in [n] \setminus S} \neg x_i(p)\right)
 \left(\bigwedge_{j \neq i}\ \bigwedge_{p \in A_{ij}(S)} \neg x_j(p)\right).
\]
for appropriately chosen sets $A_{ij}(S)$: $i, j\in [n], S\subseteq [n], |S|=k$. 
Each term $T_{i,S}$ specifies that block $i$ has ones exactly at the positions in the set $S$, and requires selected positions, described by sets $A_{ij}(S)$, in every other block to be zero. We choose these sets so that different terms are incompatible, while each term fixes only $O(n)$ bits.

The construction of sets $A_{ij}(S)$ is probabilistic. 
Let $M=8n$.
For each $j \in [n]$, we choose $P_j \in [n]^M$ uniformly at random. Let $P_j(a)$ denote the $a^{\rm th}$ entry.  
For $S \subseteq [n]$, we define $m_j(S) = \min\{a \in [M] : P_j(a) \in S\}$ when this set is nonempty, and define $m_j(S) = M + 1$ otherwise.

\begin{lemma}\label{lem:first-hit-sum}
Let $P_1, \ldots, P_n$ be independent and uniformly random. With probability $1-o(1)$, simultaneously for all
$S\subseteq[n]$ with $|S|\ge n/3$, we have 
\[
 \sum_{j = 1}^n m_j(S) \leq 8 n .
\]
\end{lemma}

We now fix a choice of $P_1, \ldots, P_n$ that satisfies the property of \Cref{lem:first-hit-sum}.
For each pair of distinct $i,j$ and $|S| \geq n / 3$, define the prefix set
\[
 A_{ij}(S) = \{P_i(a) : 1 \leq a \leq m_j(S)\}.
\]
The list $P_j$ determines the prefix length $m_j(S)$, and the first $m_j(S)$ entries of $P_i$ specify the positions required to be zero in block $j$. Since the set $A_{ij}(S)$ consists of the distinct entries among $P_i(1),\ldots,P_i(m_j(S))$, we have $|A_{ij}(S)| \leq m_j(S)$.
Also,
\[
\lvert T_{i,S}\rvert
= |S|+\bigl(n-|S|\bigr)
  +\sum_{j\in[n]\setminus\{i\}}\lvert A_{ij}(S)\rvert 
  \le n+\sum_{j\in[n]\setminus\{i\}}m_j(S) \le n+\sum_{j=1}^{n}m_j(S)
\le 9n,
\]
where the last inequality follows from \Cref{lem:first-hit-sum}, since
$|S|=k=\lfloor n/2\rfloor\ge n/3$.

We show that the DNF $f$ is unambiguous.

\begin{lemma}\label{lem:dnf-unambiguous}
For each $x\in\{0, 1\}^{n^2}$, at most one term in $f(x)$ is true.
\end{lemma}

\begin{proof}
Two terms with the same block index $i$ but different sets $S,S'$ require different assignments in that block.  Consider $T_{i,S}$ and $T_{j,S'}$ with $i \neq j$ and suppose $m_j(S) \leq m_i(S')$.  The position $p = P_j(m_j(S))$ belongs to $S$ and occurs among the first $m_i(S')$ entries of $P_j$, so $p \in A_{ji}(S')$.  Thus $T_{i,S}$ requires $x_i(p) = 1$, while $T_{j,S'}$ requires $x_i(p) = 0$.  Interchanging the two terms handles the opposite inequality.
\end{proof}

We now define the partial function $H(x, y)$.
If $x$ satisfies the term corresponding to $(i, S)$, we define
\begin{equation}\label{eq:partial-function}
 H(x,y) = \bigvee_{p \in S} y_i(p).
\end{equation}
We set $H(x,y) = *$ (undefined) when $f(x) = 0$, so the domain of $H$ consists of the pairs $(x,y)$ with $f(x) = 1$.

We claim that the function has two properties. First, if $H(x, y)=\beta, \beta\in\{0, 1\}$, then this can be certified via a certificate description $z$ consisting of $O(n \ln n)$ bits. Moreover, a description $z$ can be checked by verifying a number of conditions each of which requires querying only $O(\ln n)$ variables.

\begin{lemma}\label{lem:description-verification}
There is a verification procedure $A(x, y, \beta, z)$, $z \in \{0,1\}^d$, $d = \Theta(n \ln n)$ such that
\begin{enumerate}
\item
$H(x, y)=\beta, \beta\in\{0, 1\}$ if and only if there exists $z$ such that $A(x, y, \beta, z)=1$;
\item
$A(x, y, \beta, z) = \bigwedge_{i=1}^m A_i(x, y, \beta, z)$ where $m=O(n)$ and each $A_i$ can be evaluated deterministically with $O(\ln n)$ queries to $(x,y,\beta,z)$.
\end{enumerate}
\end{lemma}

The description and checks defining $A$ are given in \Cref{app:description-verification}.

Second, if $(x, y)$ is the all-zero input, both $H(x, y)=0$ and $H(x, y)=1$ remain possible after seeing less than a constant fraction of all variables.

\begin{lemma}\label{lem:undefined-input}
Let $\tau_n = n (\lceil n / 2 \rceil + 1)$.
Every partial assignment consistent with $(x, y) = 0^{2 n^2}$ that fixes fewer than $\tau_n$ variables has completions $(x^{(0)}, y^{(0)})$ 
and $(x^{(1)}, y^{(1)})$ with 
$H(x^{(0)}, y^{(0)})=0$ and $H(x^{(1)}, y^{(1)})=1$.
\end{lemma}

The key idea for the proof 
is that if less than $\tau_n = n (\lceil n / 2 \rceil + 1)$ variables are queried in $(x, y)$,
there is a block $i$ with the property that less than $\lceil n/2 \rceil + 1$ variables are queried in $x_i$ and $y_i$ together. 
Then, we can choose $S$ so that $T_{i, S}$ is satisfied
and $\bigvee_{p \in S} y_i(p)$ can be set to any of values 0 and 1, depending on unqueried variables.

\subsection{Cheat-sheet totalization}

We now apply the standard cheat-sheet construction of \cite{AaronsonBenDavidKothari2016}. 
Put $r = \lfloor \log_2 \tau_n \rfloor + 1$. An input to $F_n$ consists of
\begin{enumerate}
\item
$r$ inputs $(x^{(1)}, y^{(1)})$, $\ldots$, $(x^{(r)}, y^{(r)})$ to function $H$;
\item
$2^r$ cheat-sheets, with each cheat-sheet consisting of $rd$ input bits.
\end{enumerate}

$F_n=1$ if $H(x^{(1)}, y^{(1)})=\beta_1, \ldots, H(x^{(r)}, y^{(r)}) = \beta_r$ and 
the $(\beta_1 \ldots \beta_r)^{\rm th}$ cheat-sheet contains 
$(z^{(1)}, \ldots, z^{(r)})$ such that, for all $i\in [r]$, $A(x^{(i)}, y^{(i)}, \beta_i, z^{(i)}) = 1$.
Otherwise, $F_n=0$.

From the standard cheat-sheet framework, we have
\[ \operatorname{Q}(F_n) = O( r \ln r \operatorname{Q}(H) + \sqrt{r n} \ln n ).\]
In \Cref{lem:query-transfer}, we will sharpen this, removing the $\ln r$ factor.

We also have

\begin{lemma}\label{lem:certificate-lower}
The all-zero input $Z_0$ to $F_n$ satisfies $\operatorname{C}_0(F_n,Z_0) = \tau_n$.
\end{lemma}

To see the lower bound part of this lemma
intuitively, consider a partial assignment of fewer than 
$\tau_n$ variables consistent with $Z_0$. By \Cref{lem:undefined-input}, both values remain possible for every copy of $H$. Moreover, since $2^r>\tau_n$, some cheat-sheet cell is completely untouched. We can choose its address, complete the copies of $H$ to the corresponding values, and fill that cell with accepting descriptions. Thus the partial assignment is not a zero-certificate.

\Cref{lem:certificate-lower} immediately yields the lower bound on $\operatorname{C}(F_n)$ for \Cref{thm:main}.
It remains to prove an upper bound on $\operatorname{Q}(H)$;
\Cref{lem:query-transfer} will then transfer this bound to $F_n$.

\subsection{Comparison to construction of Pabbaraju}\label{subsec:construction-comparison}
In \cite{pabbaraju2026}, a separate bijection from block $i$ to $[n]$ is chosen for each ordered pair $(i,j)$ of distinct blocks. The quantity $m_{i,j}(S)$ is the minimum value of this bijection on $S$. In the present work, the $n$ position lists $P_j$, with repetitions allowed, are shared across blocks, and $m_j(S)$ is the least index of an entry in $S$. This replaces $n(n-1)$ bijections by $n$ lists; the proof of \Cref{lem:first-hit-sum} then uses a union bound over sets $S$ instead of pairs $(i,S)$.

The descriptions use binary counts and cumulative prefix lengths, as specified in \Cref{app:description-verification}. The position lists and description format make $F_n$ a slightly different total Boolean function from $G$ in \cite{pabbaraju2026}. The incompatibility argument for the DNF, the definition of $H$ from the satisfied term, the cheat-sheet totalization, and the argument for the certificate lower bound follow \cite{pabbaraju2026}.

The analysis of the quantum algorithm also applies to the original construction in \cite{pabbaraju2026} when its bijections are chosen so that $\sum_{j \neq i}m_{i,j}(S) \leq 8n$ for every block $i$ and every subset $S$ of that block with $|S| \geq n/3$. \Cref{app:bijection-extension} proves that such choices exist and that, with the original description format and totalization, the resulting function $G$ satisfies
\[
 \operatorname{C}(G) = \Omega(n^2),
 \qquad
 \operatorname{Q}(G) = \widetilde O(\sqrt{n}).
\]

Similarly, we believe that our randomized algorithm can be easily adapted to the original construction of Pabbaraju (with the aforementioned extension from sets $S$ with $|S|=n/2$ to sets with $|S|\geq n/3$) and produce the same asymptotic gap.

\section{Main algorithmic ideas}

{\bf Finding a source.}
We cast the problem of finding the clause $T_{i, S}$ that is true in $f(x)$ as a problem of finding a source in a directed complete graph.

Let $S_i$ be the set of locations of 1s in the $i^{\rm th}$ block:
$S_i = \{ a : x_i(a) = 1\}$.
For distinct $i,j \in [n]$, define $m_{ij} = m_j(S_i)$, so that $m_{ij}$ is the first position in the list $P_j$ that points to a one in block $i$.

If $x$ satisfies $T_{i,S}$, then $S_i=S$ and
$m_{ij}=m_j(S)$. Moreover, the zero constraints in $T_{i,S}$
give
\[
    x_j(P_i(a))=0
    \qquad
    \text{for every $a\le m_j(S)$}.
\]
Consequently,
\[
    m_{ji}>m_j(S)=m_{ij}.
\]

We now define a directed complete graph $\mathcal G$ on $1, \ldots, n$ by defining that
$i \rightarrow j$ if $m_{ji}>m_{ij}$, or if $m_{ji} = m_{ij}$ and $i < j$. If $T_{i, S}$ is true, then $i\rightarrow j$ for all $j \neq i$. That is, $i$ is the source of $\mathcal G$.

{\bf Sparse and dense vertices.}
We also classify the vertices of $\mathcal G$ into sparse and dense vertices. The classification is performed by random sampling. We require that, with a high probability:
\begin{itemize}
\item 
All vertices $i$ with $|S_i|\leq n/3$ are classified as sparse;
\item 
All vertices $i$ with $|S_i| = k$ are classified as dense.
\end{itemize}

For dense vertices, \Cref{lem:first-hit-sum} is used to estimate the time necessary to check the existence of edges $i\rightarrow j$. For sparse vertices,
random sampling is used to determine the sparseness and eliminate them from consideration.

\section{Randomized algorithm for \texorpdfstring{$f(x)$}{f(x)}}\label{sec:randomized-algorithm}

\begin{theorem}\label{thm:randomized-upper-improved}
    Let $n$ be even. If $f(x)$ is an unambiguous DNF as defined earlier satisfying \Cref{lem:first-hit-sum}, then there is a randomized bounded-error algorithm with one-sided error computing $f(x)$ in $O(n)$ queries.
    \end{theorem}

\begin{proof}

    The main idea of the algorithm is that on average each non-satisfied block $w$ can be eliminated from consideration in two ways at least one of which on average takes a constant number of queries amortized over all blocks. One way is elimination due to low density ($<n/3$ ones), the other way -- a random element $u$ not yet eliminated has an edge $u\rightarrow w$ in $\mathcal{G}$. Thus after $n-1$ eliminations the remaining block should be satisfying. We can test that one block in $O(n)$ queries deterministically.
    
    The algorithm will maintain a current candidate $w$ for the satisfying block, a set $C\subseteq[n]$ of candidates, and a ``density counter'' $H$ for the current candidate $w$. $H$ is initially set to some positive integer constant $h$, and it is updated by sampling $x_w$ during comparison of $x_w$ with the other blocks. If $H$ ever becomes 0, $w$ is removed from $C$ and comparison with this block is halted and never resumed. Similarly, a block is removed from $C$ if it is determined that is not the unique source of the graph $\mathcal{G}$. The algorithm will select a random element $u$ from $C\setminus\{w\}$ and replace $w$ with $u$ if $m_{uw}<m_{wu}$ by comparing successive elements of $x_w$ and $x_u$ according to position lists $P_u$ and $P_w$, correspondingly. For each position, in addition a random element from $w$ is sampled and the counter $H$ is adjusted by $\pm 1$ with probabilities chosen so that density $\frac{1}{2}$ causes on average $H$ to increase, but density $<\frac{1}{3}$ -- decrease. We show in Appendix~\ref{app:density-walk} that it eliminates $i^\star$ with very small probability, whereas it very quickly eliminates sparse blocks. 
    
    Denote by $E(A)$ a uniformly random independently chosen element from set $A$. Full pseudocode for the algorithm is arranged into three parts. The entry point of the algorithm is in Algorithm~\ref{alg:randomized-improved}. Subroutine \textsc{CompareBlocks} eliminates one or both of two blocks by comparing their $m_{uw}$ and $m_{wu}$ while simultaneously trying to eliminate based on low density. Subroutine \textsc{IsSatisfied}$(i)$ checks if there exists a term $(i, S)$ that is satisfied.  Let $J$ be the randomized map
    \begin{equation*}
        J(x)=\begin{cases}
            1\text{, if }x=1\\
            1\text{, with probability }\frac{1}{7}\text{ if }x=0\\
            -1\text{, with probability }\frac{6}{7}\text{ if }x=0\\
        \end{cases}
    \end{equation*}
    where the randomness is obtained independently each time $J$ is applied.
    \begin{algorithm}[ht]
\caption{Randomized algorithm for $f(x)$}
\label{alg:randomized-improved}
\begin{algorithmic}[1]
    \State $C\gets [n]$
    \State $t\gets 0$
    \State $w\gets 0$
    \While{$t< cn$ and $|C|\geq 2$}
        \If{$w=0$}
            \State $w\gets E(C)$
            \State $H\gets h$
        \EndIf
        \State $u\gets E(C\setminus\{w\})$
        \State \textsc{CompareBlocks}$(w,u)$
        \If{$w \notin C$}
            \If{$u \in C$}
                \State $w\gets u$
                \State $H\gets h$
            \Else
                \State $w\gets 0$
            \EndIf
        \EndIf
    \EndWhile
    \If{$C=\{z\}$}
        \State $w \gets z$
    \EndIf
    \If {$w\neq 0$ and \textsc{IsSatisfied}$(w)$}
        \State \Return 1
    \Else
        \State \Return 0
    \EndIf
\end{algorithmic}
\end{algorithm}

\begin{algorithm}[ht]
\caption{Subroutine \textsc{CompareBlocks}$(w,u)$}
\label{alg:randomized-compare-blocks}
\begin{algorithmic}[1]
    \State $p\gets 1$
    \While{$p\leq M$ and $\{w,u\}\subseteq C$}
        \State $H\gets H+J(x_w(E([n])))$
        \If{$H=0$}
            \State $C\gets C\setminus\{w\}$
        \EndIf
        \If{$x_w(P_u(p))=1$}
            \State $C \gets C\setminus\{u\}$
        \EndIf
        \If{$x_u(P_w(p))=1$}
            \State $C \gets C\setminus\{w\}$
        \EndIf
        \State $t\gets t+3$
        \State $p\gets p+1$
    \EndWhile
    \If {$p>M$ and $\{w,u\}\subseteq C$}
        \State $C\gets C\setminus\{w,u\}$
    \EndIf
\end{algorithmic}
\end{algorithm}

\begin{algorithm}[ht]
\caption{Subroutine \textsc{IsSatisfied}$(i)$}
\label{alg:issatisfied}
\begin{algorithmic}[1]
    \State Query all variables in block $i$: $x_i(1), \ldots, x_i(n)$
    \If{$\sum_{p\in[n]}{x_i(p)}\neq n/2$}
        \State \Return 0
    \Else
        \For{$j\in [n]\setminus \{i\}$}
            \For{$p\in A_{ij}(S_i)$}
                \If{$x_j(p)=1$}
                    \State \Return 0
                \EndIf
            \EndFor
        \EndFor
        \State \Return 1
    \EndIf
\end{algorithmic}
\end{algorithm}  

    One iteration of the \textbf{while} loop in subroutine \textsc{CompareBlocks} will be referred to as a step of the algorithm. In each step 3 queries are used.  Note that the algorithm always stops after $O(n)$ queries and, furthermore, if $f(x)=0$ it will always output 0. The $t\leq cn$ limits the number queries in the outer \textbf{while} loop to $3(c+2)n$ and \textsc{IsSatisfied} uses at most $n+M=O(n)$ queries due to Lemma~\ref{lem:first-hit-sum}. The rest of the analysis deals with the $f(x)=1$ case. In this case there exists a unique satisfied term of $f$; let $(i^\star, S^\star)$ denote this term.
    \begin{lemma}\label{lem:elimination-probability}
        $\Pr[i^\star\text{ is eliminated}]\leq(3/4)^{h}$.
    \end{lemma}
    \begin{proof}
        $i^\star$ can only be eliminated due to $H$ becoming $0$. Let $X_k$ be the value of $J(x_w(E([n])))$ in the $k$-th time the algorithm density-samples block $i^\star$ and let $Y_k=h+\sum_{i=1}^k{X_i}$ be the intermediate values of $H$. Denote by $\tau_h$ the smallest $k$ such that $Y_k= 0$. Then $\Pr[i^\star\text{ is eliminated}]\leq \Pr[\tau_h<\infty]$. Now Lemma~\ref{lem:density-stopping-probability} applies with the same $X_k$, $Y_k$, and $p=\frac{|S_{i^\star}|}{n}+\frac{n-|S_{i^\star}|}{n}\cdot\frac{1}{7}=\frac{4}{7}$.
    \end{proof}
    
    For the moment assume that there is no query limit. We divide the operation of the algorithm into two phases. The first phase is all work until $i^\star$ is selected as $u$ or $w$, or $i^\star$ remains the sole candidate, and the second phase is after. Now we will estimate the average number of steps $T$ used in the first phase.  We split the costs into two contributions:
    \begin{equation}
        \mathbb{E}[T]
        \leq 
        \mathbb{E}
        \left[
        \sum_{i\in [n]}{L_i}+\sum_{m\in \{0,\ldots, n-2\}}{D_m} 
        \right]
    \end{equation}
    where 
    \begin{equation*}
        L_i=\begin{cases}\text{\# of steps used with }w=i, \text{ if }|S_i|< \frac{n}{3}\\0,\text{ otherwise}\end{cases}
    \end{equation*}
    is the contribution from sparse blocks as the candidate $w$. Define $U=C\setminus\{u,w\}$. For $m\in \{0,\ldots, n-2\}$ let
    \begin{equation*}
        D_m=\begin{cases}\text{\# of steps used while }|U|=m \text{ if }|S_w|\geq \frac{n}{3}\text{ and }i^\star \in U\\0,\text{ otherwise}\end{cases},
    \end{equation*}
    be the contribution from dense blocks that don't involve $i^\star$. In the above quantities $L_i$ and $D_m$ steps are counted and not queries. 
    
    First we show that for sparse blocks, the counter hits 0 in $O(h)$ updates in expectation, thus bounding the comparison steps used while $w$ is sparse.
    \begin{lemma}
        For all $i\in[n]$: $\mathbb{E}[L_i]= O(h)$.
    \end{lemma}
    \begin{proof}
        Let $X_k$, $Y_k$ and $\tau_h$ be as for the $i^\star$ block. Apply Lemma~\ref{lem:expected-stopping-time} with $p=\frac{|S_i|}{n}+\frac{n-|S_i|}{n}\cdot\frac{1}{7}<\frac{3}{7}$ obtaining $\mathbb E[\tau_h]\leq 7h$. Thus the expected total number of steps that sample $i$-th block is at most: $7h$.
    \end{proof}

    \begin{lemma}
        For all $m\in\{0, \ldots, n-2\}$: $\mathbb{E}[D_m]=O(1)$.
    \end{lemma}
    \begin{proof}
        Firstly, $D_0=0$. In the following assume that $m\geq 1$.
        
        To upper bound $D_m$, let $A_m$ be the event that $i^\star\in U$ and $|S_w|\geq n/3$. Since blocks are selected for comparison in uniformly random order, $\Pr[A_m]\leq \frac{m}{n}$. Just before selecting $u$ for comparison, conditioning on $A_m$ there are $m$ choices for $u$. Cost of comparing $w$ with $u$ is at most $m_{wu}$. Thus $\mathbb{E}[D_m \mid A_m]\leq \sum_{u \in [n]\setminus \{w\}}{\frac{m_{wu}}{m}} \leq \frac{8n}{m}$ where the last inequality uses Lemma~\ref{lem:first-hit-sum}. Since $D_m=0$ when $A_m$ does not occur:
        \begin{align*}
            &\mathbb{E}[D_m] = \Pr[A_m]\mathbb{E}[D_m \mid A_m] \leq \frac{m}{n}\cdot\frac{8n}{m}=8.
        \end{align*}
        \end{proof}
    Since there are at most $n$ valid $m$ values for which $D_m>0$, $\mathbb{E}[\sum_m{D_m}]=O(n)$. This concludes the proof that $\mathbb{E}[T]=O(n)$, provided $f(x)=1$. Choose a constant $c_0$ so that $\mathbb{E}[T]\leq c_0n$. 
    
    Once $i^\star$ is selected and not eliminated due to the density test, it takes at most $\sum_{u\in[n]\setminus\{i^\star\}}{m_{i^\star u}}\leq 8n$ steps to eliminate the other blocks. Thus running the algorithm for $10c_0n+8n=O(n)$ steps finds $i^\star$ and eliminates the other blocks with probability 
    \begin{equation*}
        \Pr[\text{alg. outputs 1}\mid f(x)=1] \geq \Pr[T\leq 10c_0n]-\Pr[i^\star\text{ is eliminated}]\geq 9/10-(3/4)^h
    \end{equation*}
    where the first probability is bounded by the Markov's inequality, and the second -- by Lemma~\ref{lem:elimination-probability}. Setting $h=10$ produces a success probability of $\geq 4/5$. Finally, \textsc{IsSatisfied} check takes an additional $O(n)$ queries.

\end{proof}

\section{Bounded-error quantum query complexity}\label{sec:source-proof}
\subsection{Overview}\label{subsec:block-graph}

We aim to replace classical search for a source in a similar graph by a quantum search. Our main algorithm is described as procedure \textsc{FindSource}. 

\begin{algorithm}[ht]
\caption{The procedure \textsc{FindSource}}
\label{alg:sourcefinding}
\begin{algorithmic}[1]
\Procedure{FindSource}{}
  \State Choose the table of random positions used by the test
  \Statex Throughout the procedure, including within \textsc{Search}, return $1$ before a classification or membership circuit or its inverse would make the total query count exceed $T_{\max}$.
  \State Choose $w$ uniformly from $[n]$, set $W \gets (w)$, and store the classification of $w$
  \While{$|W| < n$}
    \State $v \gets \mathrm{Search}(U(W),T_2)$
    \If{$v = \bot$}
      \State \Return $w_{|W|}$
    \EndIf
    \If{$v \notin [n] \setminus W$}
      \State \Return $1$
    \EndIf
    \State Append $v$ to $W$ and store its classification
  \EndWhile
  \State \Return $w_{|W|}$
\EndProcedure
\end{algorithmic}
\end{algorithm}

It maintains an ordered list $W = (w_1,\ldots,w_h)$ of distinct vertices of $\mathcal G$. For this list, we define the candidate set
\begin{equation}\label{eq:candidate-set}
 U(W) = \{v \in [n] \setminus W : v \to w \text{ for every } w \in W\},
 \qquad U(\emptyset) = [n].
\end{equation}
If $v \in U(W)$ and $W' = (w_1,\ldots,w_h,v)$, then
\begin{equation}\label{eq:candidate-update}
 U(W') = \{u \in U(W) : u \to v\} \subseteq U(W) \setminus \{v\}.
\end{equation}
If $s$ is a source, it belongs to $U(W)$ until it is appended to $W$. Once appended, no vertex points to it, so the next candidate set is empty. Therefore repeatedly selecting a vertex from $U(W)$, appending it to $W$, and returning the last appended vertex when the updated candidate set is empty returns $s$.

A search in $U(W)$ uses a membership test that classifies $v$ as sparse or dense and compares it with the stored vertices in order, stopping when it finds $v = w_d$ or $w_d \to v$. 

\paragraph{Classification test.}
The bound in \Cref{lem:first-hit-sum} applies to blocks with at least $n / 3$ ones. To use it throughout the algorithm, put $s_n = \lceil 288 \ln(12 n) \rceil$ and choose $s_n$ positions in each block independently and uniformly with replacement. The table of positions is fixed throughout \textsc{FindSource}. A block is declared dense if at least $3s_n / 8$ of its bits at these positions are one, and declared sparse otherwise. Its classification is computed when needed, using $s_n$ queries; write $\gamma_i = 1$ for a block declared dense and $\gamma_i = 0$ otherwise.

In \Cref{app:block-classification}, we show that, with probability at least $11 / 12$, every block with $k$ ones is declared dense and every block declared dense has more than $n / 3$ ones. 

For \textsc{FindSource}, we modify $\mathcal G$ by directing edges from blocks declared dense to blocks declared sparse, and edges between two blocks declared sparse from the smaller index to the larger. The block of a satisfied term remains the source whenever the classification test succeeds.

\paragraph{Direction test.}
For two blocks declared dense by the classification, we compare $m_{v w}$ with $m_{w v}$ by searching successively doubled prefixes, ending at length $M$. 

We now define parameters describing the complexity of this test.
Choose an absolute constant $c_0\geq1$ so that the circuit in \Cref{lem:prefix-search}, including the simulation in \Cref{sec:edge-proof}, uses at most $\lceil c_0\sqrt{L\ln n}\rceil$ queries to $x$ for every $L\in[M]$.
For distinct $u,v\in[n]$, let $\Delta_{uv} = \min\{m_{uv}, m_{vu}\}$. If $u$ and $v$ are declared dense and $u\neq v$,
let $\theta_{uv} = \left\lceil c_0\sqrt{\Delta_{uv}\ln n}\right\rceil$.
Otherwise, $\theta_{uv}=0$.

\begin{lemma}
There is a subroutine that determines whether $u\rightarrow v$ for $u, v$ declared dense by classification, 
using $O(\theta_{uv})$ queries when all prefix comparison answers are correct, and achieving a probability of correct answer at least $1-O(\frac{1}{n^5})$. 
\end{lemma}

\paragraph{Variable-time search.}
As $\Delta_{uv}$ depends on the vertex $u$ that is being tested, the test for $u\in U(W)$ has different time complexities for different $u$. We use variable-time search (in the form described in Ambainis, Kokainis and Vihrovs \cite{AmbainisKokainisVihrovs}, with some slight modifications described in \cref{app:variable-search}).

When the classification test succeeds, the complexity of variable-time search relies on the property
\begin{equation}
\label{eq:edge-square-sum}    
\max_u \sum_{v} \theta^2_{uv} \leq K n \ln n
\end{equation} 
which can be obtained from Lemma \ref{lem:first-hit-sum}.

\subsection{Complexity bounds for membership and variable-time search}\label{subsec:candidate-sets}

Let $t_v(W)$ be the sum of the query counts of the circuits for testing if $v\in U(W)$ until the membership test returns zero or one (for the sequence of tests implemented if all tests give correct answers). For a fixed input and table, these counts are deterministic (by the construction in \cref{subsec:membership}). In \Cref{subsec:membership}, we show:

\begin{lemma}\label{lem:membership-counts}
For every input and table of positions, one has $t_v(\emptyset) = s_n$ for every $v \in [n]$. Let $W = (w_1,\ldots,w_h)$ be an ordered list of distinct vertices, let $w \in U(W)$, and put $W' = (w_1,\ldots,w_h,w)$. Then, for every $v \in [n]$,
\begin{equation}\label{eq:membership-completion-update}
 t_v(W') \geq t_v(W) \quad\text{and}\quad t_v(W') - t_v(W) = O(\theta_{vw}).
\end{equation}
The difference is zero for $v \notin U(W)$ and for $v = w$.
\end{lemma}

For each integer $q \geq 0$, the checking circuit $\mathcal M_q$ acts on the list register, the register holding the tested vertex, and the response register. On a basis state where the list register encodes $W$ and the tested vertex is $v$, it reports $*$ when $q < t_v(W)$ and otherwise reports the membership bit of $v$ in $U(W)$. It preserves the list and tested vertex and adds the encoded answer to the response register. We count $2 q$ queries for each application of $\mathcal M_q$ or its inverse, matching the query count of its approximate implementation $\widetilde{\mathcal M}_q$ constructed in \Cref{subsec:membership}. The search analysis uses the exact circuits $\mathcal M_q$; \textsc{FindSource} uses their approximations $\widetilde{\mathcal M}_q$.

For fixed $W$, the circuits $\mathcal M_q$ preserve the subspace in which the list register encodes $W$. Restricted to this subspace, they have the checking form of \cite[Section~2]{AmbainisKokainisVihrovs}, with $G = U(W)$ and $t_v = t_v(W)$. More generally, $\mathrm{Search}(G,T_2)$ is defined for any set $G \subseteq [n]$ with fixed, possibly unknown stopping query counts $t_v \geq 1$. It combines searches with different numbers of stages, following \cite[Algorithms~1 and~2]{AmbainisKokainisVihrovs}. The construction and proof of the following lemma are given in \Cref{app:variable-search}.

\begin{lemma}\label{lem:full-search}
For $T_2 \geq \sqrt{n}$, every index returned by $\mathrm{Search}(G,T_2)$ belongs to $G$. Its maximum query count is $O(T_2 \ln n)$, and it returns $\bot$ with certainty when $G$ is empty. Suppose $1 \leq m = |G| < n$ and $\sum_v t_v^2 \leq T_2^2$. Then its expected query count is
\begin{equation}\label{eq:search-expectation}
 O\left(\frac{T_2}{\sqrt{m}}\left(1 + \ln \frac{n}{m}\right)\right),
\end{equation}
and an absolute constant $C \geq 1$ satisfies, for every $v \in G$,
\begin{equation}\label{eq:search-distribution}
 \Pr[\mathrm{Search}(G,T_2) = v] \leq \frac{C}{m}.
\end{equation}
The probability that $\mathrm{Search}(G,T_2)$ returns $\bot$ is at most $\delta_m$, where the numbers $\delta_m$ depend only on $m$ and satisfy
\begin{equation}\label{eq:false-emptiness}
 \sum_{m \geq 1} \delta_m < \frac{1}{100}.
\end{equation}
\end{lemma}

We next use \eqref{eq:search-distribution} to bound the number and sizes of the candidate sets. Starting from $W_0 = \emptyset$ and $U_0 = [n]$, consider a process that at each reached step with $U_h = U(W_h) \neq \emptyset$ may stop or select a vertex of $U_h$ and append it to $W_h$. Given the preceding choices, each vertex has selection probability at most $C / |U_h|$. The process stops when the candidate set is empty, and all sets after stopping are empty. Put $m_h = |U_h|$, let $R = |\{h : m_h > 0\}|$, and let $W_{\mathrm{end}}$ be the final list.

For every $b \geq 0$, there are $O(b)$ vertices whose selection decreases the number of candidates by at most $b$. Together with the bound $C / |U_h|$ on each selection probability, this gives an expected decrease $\Omega(|U_h|)$, which underlies the first two estimates below. The full definition and proof are given in \Cref{app:candidate-estimates}.

\begin{lemma}\label{lem:candidate-estimates}
For the candidate selections described above, starting from $U_0 = [n]$, the following estimates hold:
\begin{equation}\label{eq:completion-profile}
 \mathbb E R = O(\ln n), \quad \mathbb E \sum_{h:m_h > 0} m_h^{-1 / 2} = O(1), \quad \mathbb E \sum_{v = 1}^n t_v(W_{\mathrm{end}})^2 = O(n \ln^2 n).
\end{equation}
The first two estimates follow from the assumptions on the candidate selections. The third additionally assumes \eqref{eq:edge-square-sum}.
\end{lemma}

The first two estimates bound the expected query counts for classifications and repeated searches; the last estimate in \eqref{eq:completion-profile} justifies the choice of $T_2$ below.

\subsection{Analysis of FindSource}\label{subsec:source-algorithm}

We choose sufficiently large absolute constants $K_2,K_T$ and put
\begin{equation}\label{eq:source-parameters}
 T_2 = \left\lceil K_2 \sqrt{n} \ln n\right\rceil, \qquad
 T_{\max} = \lceil K_T T_2 \ln n\rceil.
\end{equation}
\textsc{FindSource} (\Cref{alg:sourcefinding}) uses $\widetilde{\mathcal M}_q$ and its inverse in each search. For the stopping rule, the circuit computing a classification and its inverse each use $s_n$ queries, while $\widetilde{\mathcal M}_q$ and its inverse each use $2 q$ queries.

We first analyze the procedure obtained by replacing $\widetilde{\mathcal M}_q$ and its inverse with $\mathcal M_q$ and its inverse. \Cref{lem:source-implementation} bounds the resulting change in output probabilities.

\begin{lemma}\label{lem:source-implementation}
For every input and fixed table of positions, each output probability differs by $o(1)$ between the procedure using $\mathcal M_q$ and its inverse and the procedure using $\widetilde{\mathcal M}_q$ and its inverse, uniformly over the input and table. The unitary circuit for \textsc{FindSource} makes exactly $T_{\max}$ queries on every input. Measuring its output register gives the output distribution of \textsc{FindSource}, and reversing its gates gives the inverse circuit.
\end{lemma}

The proof of \Cref{lem:source-implementation} is given in \Cref{app:source-implementation}. We can now analyze the whole procedure $\mathrm{FindSource}$.

\begin{lemma}\label{lem:block-finding}
For all sufficiently large $n$, if $x$ satisfies the term $(i,S)$, then $\mathrm{FindSource}$ returns the block index $i$ with probability at least $5 / 6$, where the probability is over the random positions and the measurement outcomes of the procedure. Its output lies in $[n]$ on every input. The unitary circuit constructed in \Cref{app:source-implementation} has a fixed query count $O(\sqrt{n} \ln^2 n)$ on every input; measuring its output register gives the output distribution of $\mathrm{FindSource}$.
\end{lemma}

\begin{proof}
Fix $x$ satisfying $(i,S)$ and first fix a table for which the classification test succeeds. Then $i$ is the source of $\mathcal G$ and \eqref{eq:edge-square-sum} holds. Analyze \textsc{FindSource} using $\mathcal M_q$ without the stopping rule at $T_{\max}$. For the analysis only, the run halts after a vertex is appended and classified if $\sum_v t_v(W)^2 > T_2^2$, including after the initial choice. Let $W_{\mathrm{end}}$ be the final list, including any vertex whose append caused this stop, and let $T'$ be the total query count of this run.

The additional stopping rule is tested only between calls to \textsc{Search}. Each call uses a list $W$ satisfying $\sum_v t_v(W)^2 \leq T_2^2$ and is executed until it returns an output. Fix the outcomes preceding a call. If $U(W) \neq \emptyset$, \Cref{lem:full-search} bounds the probability of returning each fixed vertex in $U(W)$ by $C / |U(W)|$. The initial choice is uniform, and a search on an empty set returns $\bot$ with certainty. Thus \Cref{lem:candidate-estimates} applies to the sequence of selections, including the last appended vertex: its selection obeys the same probability bound, and \eqref{eq:membership-completion-update} bounds the resulting increases in $t_v(W)$.

For the final list of selected vertices $W_{\mathrm{end}}$, the value $t_v(W_{\mathrm{end}})$ is the query count of the membership test for $v$ with correct comparison answers. By \eqref{eq:membership-completion-update}, $t_v(W) \leq t_v(W_{\mathrm{end}})$ at every earlier step. Markov's inequality and the last estimate in \eqref{eq:completion-profile} give
\begin{equation}\label{eq:completion-profile-tail}
 \Pr\left[\sum_{v = 1}^n t_v(W_{\mathrm{end}})^2 > T_2^2\right] \leq \frac{1}{100}
\end{equation}
for sufficiently large $K_2$. After the initial choice, the positive sizes of the candidate sets are distinct and less than $n$. Suppose the outcomes preceding a search with $m > 0$ candidates are fixed. The probability that this search returns $\bot$ is then at most $\delta_m$. A union bound and \eqref{eq:false-emptiness} therefore bound the probability that any search returns $\bot$ on a nonempty candidate set by $1 / 100$.

For the query count, recall that $U_h = U(W_h)$ is the candidate set at step $h$ of the analyzed run and $m_h = |U_h|$. For each search with $m_h > 0$, \eqref{eq:search-expectation} bounds the expected query count, with the earlier outcomes fixed, by $O(T_2 \ln n / \sqrt{m_h})$. There is at most one search on an empty set, using $O(T_2 \ln n)$ queries. At most $R$ vertices are classified, each using $s_n = O(\ln n)$ queries. Applying the first two estimates in \Cref{lem:candidate-estimates} gives
\begin{equation}\label{eq:completed-search-query-count}
 \mathbb E T' = O\left(s_n \mathbb E R + T_2 \ln n\left(1 + \mathbb E \sum_{h:m_h > 0} m_h^{-1 / 2}\right)\right) = O(T_2 \ln n).
\end{equation}
Markov's inequality gives $\Pr[T' > T_{\max}] \leq 1 / 100$ for sufficiently large $K_T$.

Compare the run of \textsc{FindSource} using $\mathcal M_q$ and enforcing the maximum query count $T_{\max}$ with the analyzed run, using the same initial vertex and measurement outcomes whenever both runs execute the same circuit. Before the analyzed run reaches its additional stopping rule, a first difference can occur only if the run enforcing $T_{\max}$ returns $1$ before a circuit would make its query count exceed $T_{\max}$. The analyzed run completes that circuit, implying $T' > T_{\max}$. With probability at least $24 / 25$, neither stopping rule is reached and every search on a nonempty set returns a vertex. On this event the runs agree.

The source remains a candidate until selected, and each selection strictly decreases the candidate set by \eqref{eq:candidate-update}. Thus the source is eventually appended. It is then the last entry of the list and the candidate set is empty, so the procedure returns it.

By \Cref{lem:source-implementation}, for every table on which the classification test succeeds, \textsc{FindSource} returns $i$ with probability at least $24 / 25 - o(1)$. Together with the $11 / 12$ probability that the classification test succeeds, this gives total success probability at least $5 / 6$ for sufficiently large $n$. The same lemma gives fixed query count $T_{\max} = O(\sqrt{n} \ln^2 n)$ on every input.
\end{proof}

\subsection{Algorithms for \texorpdfstring{$H$ and $F_n$}{H and F\_n}}

The algorithm for $H$ runs \textsc{FindSource} and then searches the returned block $i$, with error at most $1 / 100$, for a position $p$ with $x_i(p) \wedge y_i(p) = 1$. It checks a reported position and returns one only if the check succeeds; otherwise it returns zero.

\begin{lemma}\label{lem:query-upper}
The algorithm for $H$ has error at most $1 / 4$ on its domain and maximum query count $O(\sqrt{n} \ln^2 n)$ on every input. It has a unitary implementation with fixed query count $T_H = O(\sqrt{n} \ln^2 n)$. Measuring the circuit's output register gives the algorithm's output distribution, and reversing its gates gives the inverse circuit.
\end{lemma}

\begin{proof}
Suppose $x$ satisfies $(i,S)$. By \Cref{lem:block-finding}, \textsc{FindSource} returns $i$ with probability at least $5 / 6$. For this block, $x_i(p) \wedge y_i(p) = 1$ if and only if $p \in S$ and $y_i(p) = 1$.

The search in the returned block uses $O(\sqrt{n})$ queries by \cite[Theorem~3]{BuhrmanCleveDeWolfZalka}. Each query to the string $(x_i(p) \wedge y_i(p))_{p \in [n]}$ is implemented with four queries to $(x,y)$: the circuit queries $x_i(p)$ and $y_i(p)$, adds their conjunction to the response bit, and uncomputes the two queried bits while leaving the register containing $i$ unchanged. Hence the algorithm for $H$ has error at most $1 / 6 + 1 / 100 < 1 / 4$ and maximum query count $O(\sqrt{n} \ln^2 n)$.

By \Cref{lem:block-finding}, the circuit for \textsc{FindSource} stores the returned block index in its output register. This register controls the subsequent search without being measured. The search has fixed query count $O(\sqrt{n})$ and stores its outcomes as in \Cref{app:implementations}. The resulting circuit has fixed query count $T_H = O(\sqrt{n} \ln^2 n)$ on every input and the stated output distribution.
\end{proof}

The following reduction then gives the query bound for $F_n$.
\begin{lemma}\label{lem:query-transfer}
Suppose an algorithm for $H$ has error at most $1 / 4$ on its domain and has a unitary implementation making exactly $T_H$ queries on every input. Then
\[
 \operatorname{Q}(F_n) = O\left(r T_H + \sqrt{n r} \ln n\right).
\]
\end{lemma}

When the $r$ inputs lie in the domain of $H$, the recovery algorithm of \cite[Theorem~3]{BuhrmanNewmanRohrigDeWolf} recovers the address using the circuit for $H$ and its inverse $O(r)$ times. The addressed cell has $O(n r)$ checks of $O(\ln n)$ queries each, so searching for a rejected check gives the second term. The proof of \Cref{lem:query-transfer} is given in \Cref{app:query-transfer}.

The algorithm for $F_n$ recovers an address using this circuit for $H$, checks that cell, and returns its acceptance bit. By \Cref{lem:query-transfer} and $r = \Theta(\ln n)$,
\[
 \operatorname{Q}(F_n) = O\left(r T_H + \sqrt{n r} \ln n\right) = O(\sqrt{n} \ln^3 n).
\]

\section{Exact and zero-error quantum query complexity}\label{sec:exact-query}

The exact algorithm first finds the satisfied DNF term in each input to $H$, or determines that none is satisfied. \textsc{FindTerm} (\Cref{alg:exact-term}) performs this search; \textsc{Evaluate} (\Cref{alg:exact-evaluation}) uses it to compute $F_n$. A restriction of the count fields gives the matching lower bound for zero-error algorithms in \Cref{thm:exact-zero-error}.

\subsection{A randomized algorithm for finding a satisfied term}

To find a satisfied term of $f$, \textsc{Eliminate} (\Cref{alg:elimination}) compares pairs of blocks. If $(i,S)$ is satisfied, its zero constraints give $m_{ij} < m_{ji}$ for every $j \neq i$. A comparison with block $j$ that reaches position $m_j(S)$ therefore keeps $i$.

To compare an ordered pair of blocks $(i,j)$ at prefix length $\ell$, the procedure queries $x_i(P_j(t))$ and $x_j(P_i(t))$ for $1 \leq t \leq \min\{\ell,M\}$. It keeps the block with the smaller least nonzero index, assigning index $\ell + 1$ to a prefix containing only zeros. A tie is resolved in favour of the first block in the ordered pair.

Let $L = \lceil\log_2 n\rceil$ and $N = 2^L$. The labels $n + 1,\ldots,N$ are auxiliary candidates. A block in $[n]$ is kept when compared with an auxiliary label; between two auxiliary labels, the first is kept. These comparisons use no queries. The procedure records its permutations, candidate sets, query answers, and verification result; write $\omega$ for this record.

\begin{algorithm}[ht]
\caption{Elimination and verification}\label{alg:elimination}
\begin{algorithmic}[1]
\Procedure{Eliminate}{$x$}
  \State $U_0 \gets [N]$
  \For{$a = 0,\ldots,L - 1$}
    \State Choose a fresh uniform permutation of $U_a$ and pair consecutive entries
    \State Compare each pair at prefix length $32 \cdot 2^a$ and let $U_{a + 1}$ contain the retained labels
  \EndFor
  \State Read the sole remaining block $x_i$ and let $S$ be its support
  \If{$|S| = k$}
    \State Check $x_j(P_i(t)) = 0$ for every $j \neq i$ and $1 \leq t \leq m_j(S)$
    \State \Return $(i,S)$ if all checks accept
  \EndIf
  \State \Return $\bot$
\EndProcedure
\end{algorithmic}
\end{algorithm}

At least one label in $[n]$ remains after each round, so the final label specifies a block. The verification tests all literals of $T_{i,S}$, and hence every returned term is satisfied. On a positive input, \textsc{Eliminate} returns the satisfied term if and only if its block is retained in every round.

At round $a$, there are $N / 2^{a + 1}$ comparisons, each using at most $64 \cdot 2^a$ queries. The final verification uses at most $n + M$ queries by \Cref{lem:first-hit-sum}. The query count is at most
\begin{equation}\label{eq:elimination-query-count}
 T_n = 32 N L + n + M = O(n \ln n).
\end{equation}

\subsection{Exact and zero-error bounds for \texorpdfstring{$F_n$}{F\_n}}

We use the record of \textsc{Eliminate} to obtain a circuit whose acceptance probability is $f(x) / 4$. One iteration of amplitude amplification will then give exact search.

For an accepted record $\omega$, let $(i,S)$ be its returned term and let $U_a$ be its candidate set before round $a$. Define
\begin{equation}\label{eq:elimination-retention-probability}
 p_a(\omega) = 1 - \frac{|\{j \in U_a \setminus \{i\} : j \leq n \text{ and } m_j(S) > 32 \cdot 2^a\}|}{2(|U_a| - 1)},
 \qquad 0 \leq a < L.
\end{equation}
Fix a positive input with satisfied term $(i,S)$ and the choices before round $a$, with $i \in U_a$. The fresh uniform permutation of $U_a$ pairs $i$ with each other label with probability $1 / (|U_a| - 1)$, and either order in the pair is equally likely. For a block $j$ counted in the numerator of \eqref{eq:elimination-retention-probability}, both queried prefixes contain only zeros, so the tie rule discards $i$ when it occurs second. Every other comparison retains $i$, including comparisons with auxiliary labels. This gives the conditional retention probability $p_a(\omega)$.

Once verification returns $(i,S)$, the recorded candidate sets and the fixed lists determine all the values $p_a(\omega)$ without further queries. We assign to the accepted record the product of the reciprocal retention probabilities for its rounds:
\begin{equation}\label{eq:elimination-weight}
 \mu(\omega) = \prod_{a = 0}^{L - 1}p_a(\omega)^{-1},
\end{equation}
and set $\mu(\omega) = 0$ on rejection.

\begin{lemma}\label{lem:weighted-success}
For every input $x$, the record $\omega$ of \textsc{Eliminate} satisfies $0 \leq \mu(\omega) \leq 2$, and
\begin{equation}\label{eq:weighted-success}
 \mathbb E_\omega \mu(\omega) = f(x),
\end{equation}
where the expectation is over the permutation choices.
\end{lemma}
The proof is given in \Cref{app:exact-preparation}.

On the joint zero state, the circuit $\mathcal A_x$ prepares each permutation register in uniform superposition and performs the comparisons and verification of \Cref{alg:elimination}, storing the candidate sets, query answers, and verification result. It then applies the following rotation to a flag qubit initialized to zero:
\begin{equation}\label{eq:exact-flag-rotation}
 \ket{\omega}\ket{0}
 \longmapsto
 \ket{\omega}\left(\sqrt{1 - \mu(\omega) / 4}\ket{0} + \sqrt{\mu(\omega) / 4}\ket{1}\right).
\end{equation}
The rotation uses no queries. The circuit $\mathcal A_x$ and its inverse each use $T_n$ queries.

Let $\Pi$ project onto the subspace where the flag qubit is in state $\ket{1}$. By \Cref{lem:weighted-success},
\begin{equation}\label{eq:exact-flag-probability}
 \|\Pi \mathcal A_x\ket{0}\|^2 = \frac{1}{4}\mathbb E_\omega \mu(\omega) = \frac{f(x)}{4}.
\end{equation}
Every record obtained with flag outcome $1$ has passed verification. The following procedure uses this acceptance probability to find the term with certainty.

\begin{algorithm}[ht]
\caption{Exact search for a satisfied term}\label{alg:exact-term}
\begin{algorithmic}[1]
\Procedure{FindTerm}{$x$}
  \State Prepare $\mathcal A_x\ket{0}$
  \State Apply $\mathcal A_x(2\ket{0}\bra{0} - I)\mathcal A_x^{-1}(I - 2\Pi)$
  \State Measure the flag qubit and the record $\omega$
  \If{the flag outcome is $1$}
    \State \Return the term $(i,S)$ specified by $\omega$
  \EndIf
  \State \Return $\bot$
\EndProcedure
\end{algorithmic}
\end{algorithm}

\begin{lemma}\label{lem:exact-term}
\Cref{alg:exact-term} returns the satisfied term $(i,S)$ when $f(x) = 1$ and returns $\bot$ when $f(x) = 0$, with certainty. It uses $3 T_n = O(n \ln n)$ queries on every input.
\end{lemma}
\begin{proof}
By \eqref{eq:exact-flag-probability}, the initial acceptance probability is $1 / 4$ when $f(x) = 1$ and zero when $f(x) = 0$. The step of amplitude amplification in \Cref{alg:exact-term} gives the state $2 \Pi \mathcal A_x\ket{0}$ in the first case and leaves the flag zero in the second \cite[Eq.~(8)]{BrassardHoyerMoscaTapp}. Since every record obtained with flag outcome $1$ specifies a satisfied term, the output is correct.

The procedure applies $\mathcal A_x$ twice and its inverse once. Both reflections use no queries, giving the query count $3 T_n$.
\end{proof}

The term returned by \textsc{FindTerm} identifies the $y$-bits whose OR is the value of $H$. The following algorithm computes these values and checks the addressed cell.

\begin{algorithm}[ht]
\caption{Exact evaluation of $F_n$}\label{alg:exact-evaluation}
\begin{algorithmic}[1]
\Procedure{Evaluate}{}
  \For{$\ell = 1,\ldots,r$}
    \State Run \textsc{FindTerm} on $x^{(\ell)}$
    \If{the output is $\bot$}
      \State \Return $0$
    \EndIf
    \State Let $(i,S)$ be the returned term
    \State Query $y_i^{(\ell)}(p)$ for $p \in S$ and set $\beta_\ell \gets \bigvee_{p \in S}y_i^{(\ell)}(p)$
  \EndFor
  \State Check the descriptions in cell $(\beta_1,\ldots,\beta_r)$ using \Cref{lem:description-verification}
  \State \Return $1$ if every check accepts, and $0$ otherwise
\EndProcedure
\end{algorithmic}
\end{algorithm}

\begin{lemma}\label{lem:exact-total-upper}
\Cref{alg:exact-evaluation} computes $F_n$ with certainty using $O(n \ln^2 n)$ queries.
\end{lemma}
\begin{proof}
If a search returns $\bot$, the corresponding input to $H$ lies outside its domain by \Cref{lem:exact-term}, so $F_n = 0$. Otherwise, the returned terms identify the $y$-bits whose ORs are the $r$ values of $H$; evaluating these ORs and checking the addressed cell gives $F_n$.

The term searches use at most $3 r T_n$ queries, the selected $y$-bits use at most $r n$, and the description checks use $O(r n \ln n)$ queries by \Cref{lem:description-verification}. Since $T_n = O(n \ln n)$ and $r = \Theta(\ln n)$, the total is $O(n \ln^2 n)$.
\end{proof}

For the matching lower bound, we use the count fields of the description. Each field $w_p$ has $\lceil\log_2(n + 1)\rceil$ bits and must encode $\sum_{a = 1}^p x_i(a)$, as specified in \Cref{app:description-verification}.

\begin{lemma}\label{lem:zero-error-lower}
The function $F_n$ satisfies
\[
 \operatorname{Q}_0(F_n) \geq r n \lceil\log_2(n + 1)\rceil.
\]
\end{lemma}
\begin{proof}
In each of the $r$ inputs to $H$, fix $x_1$ to have support $[k]$, all other blocks to zero, and $y = 0$. All $r$ values of $H$ are then zero, so only cell $0^r$ can accept. Choose accepted descriptions there and fix all their fields except the counts; fix every other cell as well. The remaining $m = r n \lceil\log_2(n + 1)\rceil$ bits have a unique accepting assignment $u^\star$, because each $w_p$ must encode $\min\{p,k\}$.

After complementing the coordinates where $u_j^\star = 0$, the restricted function has $1^m$ as its unique positive input and therefore equals $\operatorname{AND}_m$. Hence \cite[Proposition~6.1]{Beals2001} gives $\operatorname{Q}_0(F_n) \geq \operatorname{Q}_0(\operatorname{AND}_m) = m$.
\end{proof}

\paragraph{Acknowledgments.}
We thank Jevg{\=e}nijs Vihrovs for bringing \cite{pabbaraju2026,BenDavidKothari2026} to our attention and for pointing out the near-quartic separation between $\operatorname{R}_0$ and $\operatorname{Q}$.

\paragraph{Use of LLM tools.}
We used large language model (LLM) tools at various stages of the project to assist with literature searches, code development for early computational experiments, detecting potential logical flaws in the mathematical content, and drafting and revising the manuscript. We reviewed, corrected, and independently verified all LLM-generated output and take full responsibility for the final manuscript and all its content, including its accuracy, originality and integrity.

\bibliographystyle{plainurl}
\bibliography{bib}

\appendix
\crefalias{section}{appendix}
\crefalias{subsection}{subappendix}

\section{Two facts about biased random walks on a line}\label{app:density-walk}

Let $X_i$ for $i\in\mathbb{N}$ be independent random variables distributed as
\begin{equation*}
    X_i=\begin{cases}
        1, &\text{with probability }p\\
        -1, &\text{with probability }1-p.
    \end{cases}
\end{equation*}
For $j\geq 1$ let $Y_j=Y_0+\sum_{i=1}^j{X_i}$ where $Y_0=h$. One can think of the random variables $Y_j$ as a gambler's bankroll when playing against an infinitely rich adversary with the following rules: the gambler starts with $h$ coins and in each round wins 1 coin with probability $p$ and loses one coin with probability $1-p$. We will use and prove (well-known facts \cite[Sections XIV.2 and XIV.3]{feller1968vol1} reproduced here for completeness) that
\begin{enumerate}
    \item if the value of the game is positive, then with high probability the gambler does not go bankrupt;
    \item if the value of the game is negative, the gambler goes bankrupt in very few rounds.
\end{enumerate}
For $h \in \mathbb{N}_+$, let $\tau_h=\inf\{j:Y_j = 0\}$ and $r=\frac{1-p}{p}$. Then
\begin{lemma}\label{lem:density-stopping-probability}
    If $p > \frac{1}{2}$, then
    \[
        \Pr[\tau_h<\infty]= r^h,
    \]
    In particular, if $p\geq4/7$, then
    $\Pr[\tau_h<\infty]\leq(3/4)^{h}$.
\end{lemma}
\begin{proof}
    When $p=1$ the claim is trivial. In the following assume that $\frac{1}{2}<p<1$.
    
    First we will deal with the case where gambler stops playing once he reaches 0 or $N$ coins. Denote by $q_h^{(N)}$ the probability that starting at $h$ coins the gambler reaches $0$ before reaching $N$. Then
    \begin{equation*}
        \begin{cases}
            q_0^{(N)}=1\\
            q_N^{(N)}=0\\
            q_i^{(N)}=pq_{i+1}^{(N)}+(1-p)q_{i-1}^{(N)},\quad i\in\{1, \ldots, N-1\}
        \end{cases}.
    \end{equation*}
    This system of linear equations has one solution: $q_i^{(N)}=\frac{r^i-r^N}{1-r^N}$. Indeed, write $d_i^{(N)}=q_{i+1}^{(N)}-q_i^{(N)}$ for $i\in \{0, \ldots, N-1\}$. Then the non-boundary equations give $(1-p)d_i^{(N)}=pd_{i+1}^{(N)}$ or, equivalently, $d_{i+1}^{(N)}=r d_i^{(N)}$ for all $i\in\{0,\ldots,N-2\}$. On the other hand $-1=\sum_{i=0}^{N-1}{d_i^{(N)}}=\sum_{i=0}^{N-1}{r^id_0^{(N)}}$. Therefore $d_0^{(N)}$ and hence all other differences $d_i^{(N)}$ are uniquely determined. Differences together with boundary conditions uniquely determine $q_i^{(N)}$.

    Finally,
    \begin{equation*}
        \Pr[\tau_h<\infty] = \lim_{N\rightarrow \infty}{q_h^{(N)}}=\lim_{N\rightarrow \infty}{\frac{r^h-r^N}{1-r^N}}=r^h.
    \end{equation*}
    where the first equation follows from the probability of an increasing sequence of events.
\end{proof}

\begin{lemma}\label{lem:expected-stopping-time}
    If $p<\frac{1}{2}$, then
    \[
        \mathbb{E}[\tau_h]=\frac{h}{1-2p}.
    \]
    In particular, if $p\leq3/7$, then
    $\mathbb{E}[\tau_h]\leq 7h$.
\end{lemma}
\begin{proof}
    When $p=0$ the claim is trivial. In the following assume that $0<p<\frac{1}{2}$.
    
    As in the previous lemma, consider the case, where the gambler stops playing upon reaching $N$. Denote by $E_i^{(N)}$ the quantity $\mathbb{E}[\tau_i^{(N)}]$ where $\tau_i^{(N)}$ is the number of steps until 0 or $N$ is reached. Then
    \begin{equation*}
        \begin{cases}
            E_0^{(N)}=0\\
            E_N^{(N)}=0\\
            E_i^{(N)}=pE_{i+1}^{(N)}+(1-p)E_{i-1}^{(N)}+1, \quad i\in \{1, \ldots,N-1\}
        \end{cases}.
    \end{equation*}
    For $i\in \{0, \ldots,N\}$ define $D_i^{(N)}=E_i^{(N)}-\frac{i}{1-2p}$. Then the non-boundary equations become the familiar equations from the previous lemma: $D_i^{(N)}=pD_{i+1}^{(N)}+(1-p)D_{i-1}^{(N)}$. Denoting by $d_i^{(N)}=D_{i+1}^{(N)}-D_{i}^{(N)}$ we again obtain $d_{i}^{(N)}=r^id_{0}^{(N)}$. From the boundary conditions we obtain $\sum_{i=0}^{N-1}{r^id_0^{(N)}}=-\frac{N}{1-2p}$, hence the solution is unique. Solving for $E_i^{(N)}$ we obtain
    \begin{equation*}
        E_i^{(N)}=\frac{i}{1-2p}-\frac{N}{1-2p}\cdot\frac{1-r^i}{1-r^N}.
    \end{equation*}

    Finally, $\mathbb{E}[\tau_h]=\lim_{N \rightarrow \infty}{E_h^{(N)}}=\frac{h}{1-2p}$, where the first equality follows from the inequalities $\mathbb{E}[\min\{\tau_h,N-h\}]\leq E_h^{(N)} \leq \mathbb{E}[\tau_h]$.
\end{proof}

\section{Variable-time search}\label{app:variable-search}

We use the variable-time search algorithm of \cite{AmbainisKokainisVihrovs}.
Our application requires two additional properties: a bound on the
probability of returning any particular marked index, and a failure
bound that is summable over the number of marked indices.
We derive these properties from that construction below.

Let $G \subseteq [n]$ and consider the exact checking family
$\mathcal M_q$ described in \Cref{sec:source-proof}.
For each index $v$, there is a fixed, possibly unknown stopping
query count $t_v\geq 1$.
The checking circuit reports $*$ when $q<t_v$, and otherwise
reports whether $v\in G$.
We count $2 q$ queries for each application of $\mathcal M_q$ or its inverse.

Encode $*,0,1$ as $00,01,10$, respectively. The response register of $\mathcal M_q$ is the answer register of the search. To replace its answer at parameter $q$ by the answer at $q' \geq q$, the circuit applies $\mathcal M_q$ again, returning the response register to $00$, and then applies $\mathcal M_{q'}$. This update and its inverse use $2(q + q') \leq 4 q'$ queries. When the search is applied to $G = U(W)$, its reflections leave the list register unchanged.

\subsection{Search with a fixed number of stages}\label{subsec:variable-search}

Fix $T_2 \geq \sqrt{n}$ and put $q_0 = \lceil 3 T_2 / \sqrt{n} \rceil$, $J = \lceil \log_9 n \rceil$, and $q_j = 3^j q_0$ for $0 \leq j \leq J$.
Thus $q_0 \leq 4 T_2 / \sqrt{n}$ and $q_J < 12 T_2$.
We apply \cite[Algorithm~1]{AmbainisKokainisVihrovs} with $T = n q_0^2 / 9$, so that its successive checking thresholds are
$q_0,q_1,\ldots,q_{J-1}$.
Notice that $T_2^2\leq T\leq 16T_2^2/9$.

Set $\rho=10^{-3}$.
For $h \in [J]$, let $\mathcal B_h$ be Algorithm~1 of~\cite{AmbainisKokainisVihrovs} with stage parameter $h$ and error parameter
$\rho$.
Its final checking threshold is $r_h = q_{\min\{h,J - 1\}}$. Write $R_h = \{v \in G : t_v \leq r_h\}$ for the marked indices resolved by this final check.

\begin{lemma}
\label{lem:search-stage}
For every $h\in[J]$, the procedure $\mathcal{B}_h$
returns an index in $R_h$ or $\bot$, and uses
$O(h3^h q_0)$ queries.
The probabilities $\Pr[\mathcal{B}_h=v]$ are equal over
$v\in R_h$.
Moreover, there is an absolute constant $c$ such that
\[
    \Pr[\mathcal{B}_h=v]
    \leq c\,\frac{9^{h-1}}{n}
    \qquad\text{for every }v\in G.
\]
If $1 \leq m = |G| < n$ and $\sum_v t_v^2 \leq T_2^2$, then $\ell = \lceil \log_9(n / m) \rceil$ satisfies
\[
    |R_\ell|\geq\frac{8m}{9},
    \qquad
    \Pr[\mathcal{B}_\ell=\bot]\leq\rho.
\]
\end{lemma}

\begin{proof}
The query bound follows from Algorithm~1 of~\cite{AmbainisKokainisVihrovs}; its Lemma~1 gives the success guarantee for $\mathcal B_\ell$.

For the distribution bounds, consider the procedure before its
final amplification.
By the amplitude identity in~\cite[Lemma~2]{AmbainisKokainisVihrovs},
the resolved marked indices have equal amplitudes.
Each recursive amplification step increases their amplitude
magnitudes by at most a factor of three.
Since the initial amplitudes are $1/\sqrt{n}$, the probability
of any particular marked index is at most $9^{h-1}/n$.
The final amplification and repetitions preserve equality
among these probabilities and increase this bound by at most
a constant factor, since $\rho$ is fixed.

If $\ell<J$, then
$r_\ell^2=q_\ell^2\geq 9T_2^2/m$, so fewer than $m/9$
marked indices have stopping counts exceeding $r_\ell$.
If $\ell=J$, then
$r_\ell=q_{J-1}\geq\sqrt{T}\geq T_2$,
and all marked indices are resolved.
\end{proof}

\subsection{Search with an unknown number of solutions}\label{subsec:complete-search}

Define $\mathrm{Search}(G,T_2)$ using the schedule of
\cite[Algorithm~2]{AmbainisKokainisVihrovs}.
For $j=1,\ldots,J$, round $j$ runs
$\mathcal{B}_1,\ldots,\mathcal{B}_j$ in this order,
returning the first index produced.
If every call returns $\bot$, the search returns $\bot$.
Each call uses fresh registers and independent random choices.

\begin{proof}[Proof of \Cref{lem:full-search}]
Every returned index belongs to $G$.
In particular, an empty set gives output $\bot$ with certainty.

The maximum query count through round $k$ is
\[
    O\left(
        q_0\sum_{j=1}^{k}\sum_{h=1}^{j}h3^h
      \right)
    =O(k3^kq_0).
\]
Taking $k=J$ gives the maximum query bound
$O(T_2\ln n)$, independently of the stopping counts.

Now suppose that $1\leq m=|G|<n$ and
$\sum_v t_v^2\leq T_2^2$, and put
$\ell=\lceil\log_9(n/m)\rceil$.
Every round from $\ell$ onward contains a fresh call to
$\mathcal{B}_\ell$, which fails with probability at most
$\rho$.
Consequently, round $j>\ell$ is reached with probability
at most $\rho^{j-\ell}$.
The expected query count is therefore
\[
    O\left(
        q_0\ell3^\ell+
        q_0\sum_{j=\ell+1}^{J}j3^j\rho^{j-\ell}
      \right)
    =O(\ell3^\ell q_0)
    =O\left(
        \frac{T_2}{\sqrt{m}}
        \left(1+\ln\frac{n}{m}\right)
      \right),
\]
where the first equality uses $3\rho<1$.

For the failure probability, all $J-\ell+1$ calls to
$\mathcal{B}_\ell$ must fail.
Since $J-\ell\geq\lfloor\log_9m\rfloor$, we obtain
\[
    \Pr[\mathrm{Search}(G,T_2)=\bot]
    \leq \rho^{J-\ell+1}
    \leq \rho^{1+\lfloor\log_9m\rfloor}.
\]
Thus we may take
\[
    \delta_m=\rho^{1+\lfloor\log_9m\rfloor},
    \qquad
    \sum_{m\geq1}\delta_m
    =\sum_{a\geq0}8\cdot9^a\rho^{a+1}
    =\frac{8\rho}{1-9\rho}
    <\frac{1}{100}.
\]

It remains to bound the probability of returning a fixed
$v\in G$.
We distinguish calls with $h<\ell$ from those with
$h\geq\ell$.

For $h<\ell$, there are $\ell-h$ scheduled calls to
$\mathcal{B}_h$ before round $\ell$.
The expected number from round $\ell$ onward is at most
\[
    \sum_{j=\ell}^{J}\rho^{j-\ell}
    \leq\frac{1}{1-\rho}.
\]
Conditioned on reaching any such call, its distribution is
unchanged because it uses fresh registers and randomness.
\Cref{lem:search-stage} therefore bounds the probability
of returning $v$ from a call with $h<\ell$ by
\[
    \frac{c}{n}
    \sum_{h=1}^{\ell-1}
       \left(\ell-h+\frac{1}{1-\rho}\right)9^{h-1}
    =O\left(\frac{9^{\ell-1}}{n}\right)
    =O(1/m).
\]

For $h\geq\ell$, monotonicity of $r_h$ gives
$R_h\supseteq R_\ell$, and hence $|R_h|\geq8m/9$.
Conditioned on $\mathcal{B}_h$ returning an index, its
output is uniform on $R_h$.
Its conditional probability of returning $v$ is therefore
at most $9/(8m)$.
Taking the mixture over whichever such call first returns
an index gives the same upper bound.
Combining the two ranges proves
\[
    \Pr[\mathrm{Search}(G,T_2)=v]\leq\frac{C}{m}
\]
for an absolute constant $C$.
\end{proof}

For use in \Cref{app:source-implementation}, we also count the applications of checking circuits and their inverses.
The recursive construction in~\cite{AmbainisKokainisVihrovs} makes
$O(3^h)$ applications of checking circuits or their inverses
in $\mathcal{B}_h$, since each recursive level uses three
copies of the preceding level and $O(1)$ additional checks.
The final amplification contributes only a constant factor.
Thus the complete search contains
\begin{equation}
    O\left(
        \sum_{j=1}^{J}\sum_{h=1}^{j}3^h
      \right)
    =O(3^J)
    =O(\sqrt{n})
    \label{eq:search-circuit-count}
\end{equation}
such applications, all with checking parameters below
$12T_2$.

The procedure has a unitary implementation by preparing
each random choice in a separate register and storing
measurement outcomes instead of measuring them.
Later calls are controlled on no index having yet been
returned.
All recorded outcomes and work registers are retained,
so measuring the final output register reproduces the
distribution analysed above, and reversing the circuit
gives its inverse.
\section{Supporting proofs for the quantum algorithms}\label{app:comparisons-membership}

The constructions below specify the circuits for comparing blocks, testing candidate membership, and preparing the exact search. We first state the unitary conventions used by these algorithms, then prove the comparison, candidate, and approximation estimates for \textsc{FindSource}.

\subsection{Unitary implementations}\label{app:implementations}\label{app:answer-addition}

In each unitary implementation below, the circuit prepares every random choice in its own register and copies each intermediate computational basis outcome to a separate register instead of measuring it. Each possible return has a fresh register consisting of a flag and an output field, initialized to the all-zero state $\ket{\mathrm{init}}$. When that return occurs, the circuit sets the flag to one and writes the returned value in the output field. Suppose the procedure makes at most $T$ queries for every sequence of stored values. For each $t \in [T]$, input-independent gates controlled by those values advance the computation on each corresponding subspace to its $t$th query or to a return, and move the relevant address and response data to common query registers. A controlled application of the standard oracle, implemented with one query, performs the $t$th query on every subspace where no return has been recorded and acts as the identity on the remaining subspaces. After the $T$th oracle application, the remaining input-independent gates copy the output field of the first return register not in $\ket{\mathrm{init}}$ to a designated output register initialized to zero, leaving all return registers unchanged. Measuring this output register gives the procedure's output distribution, and reversing all gates gives the inverse circuit.

Fix an integer $m \geq 1$ and strings $a_j \in \{0,1\}^m$. Let $\mathcal H_{\mathrm{index}}$, $\mathcal H_{\mathrm{answer}}$, $\mathcal H_{\mathrm{aux}}$, and $\mathcal H_{\mathrm{response}}$ be the state spaces of the index, answer, auxiliary, and response registers, respectively. The answer and response registers contain $m$ qubits. Put $\mathcal H_{\mathrm{work}} = \mathcal H_{\mathrm{answer}} \otimes \mathcal H_{\mathrm{aux}}$. The joint state space is
\[
 \mathcal H_{\mathrm{index}} \otimes \mathcal H_{\mathrm{work}} \otimes \mathcal H_{\mathrm{response}}.
\]
Let
\begin{equation}\label{eq:unitary-by-index}
 \mathcal U = \sum_j \ket{j} \bra{j} \otimes \mathcal U_j \otimes I_{\mathrm{response}},
\end{equation}
where each $\mathcal U_j$ is unitary on $\mathcal H_{\mathrm{work}}$. Define $\mathcal C \ket{j,z} = \ket{j,z \oplus a_j}$ and extend it by the identity on $\mathcal H_{\mathrm{work}}$. Let $\mathsf C \ket{a,z} = \ket{a,z \oplus a}$ add the answer to the response and act as the identity on the index and auxiliary registers. Put $\widetilde{\mathcal C} = \mathcal U^\dagger \mathsf C \mathcal U$ and
\[
 \Pi_0 = I_{\mathrm{index}} \otimes \ket{0}_{\mathrm{work}} \bra{0}_{\mathrm{work}} \otimes I_{\mathrm{response}}.
\]

\begin{lemma}\label{lem:unitary-approximation}
The unitary $\widetilde{\mathcal C}$ satisfies $\widetilde{\mathcal C}^2 = I$. If, for every $j$, measuring the answer register after applying $\mathcal U_j$ to $\ket{0}_{\mathrm{work}}$ gives $a_j$ with probability at least $1 - \delta$, where $\delta \geq 0$, then
\[
 \left\|(\widetilde{\mathcal C} - \mathcal C)\Pi_0\right\| \leq 2 \sqrt{\delta}.
\]
\end{lemma}

\begin{proof}
The equality $\mathsf C^2 = I$ gives $\widetilde{\mathcal C}^2 = I$. For fixed $j$, the component of $\mathcal U_j \ket{0}_{\mathrm{work}}$ with an incorrect answer has norm at most $\sqrt{\delta}$. On the complementary component, $\mathsf C$ and $\mathcal C$ agree; their difference has norm at most two on the incorrect component. This bound holds for every response state. Because $\mathcal U$ preserves the index register, orthogonality of the basis states $\ket{j}$ gives
\[
 \left\|(\mathsf C - \mathcal C)\mathcal U \Pi_0\right\| \leq 2 \sqrt{\delta}.
\]
Since $\mathcal C$ changes only the response for each fixed index, it commutes with $\mathcal U$. Multiplying by $\mathcal U^\dagger$ gives the claimed bound.
\end{proof}

\subsection{Comparing blocks}\label{app:block-classification}\label{app:prefix-search}\label{sec:edge-proof}\label{app:edge-implementation}

For a block of Hamming weight $k$, the expected number of ones  at the chosen positions is $s_n k/n \geq 3s_n/7$, since $n \geq 6$; for a block of Hamming weight at most $n/3$, it is at most $s_n/3$. The declaration threshold is $3s_n/8$, so either erroneous declaration requires a one-sided deviation from the corresponding expectation of at least $s_n/24$. Hoeffding's inequality therefore bounds its probability by $\exp(-s_n / 288) \leq 1 / (12 n)$.  A union bound over the blocks gives success probability at least $11 / 12$ for every fixed $x$.

Given a block index $i$, the circuit computes $\gamma_i$ by querying the bits at its $s_n$ table positions and storing them. This circuit preserves the block index and the table, and reversing its gates gives the inverse.

To compare $m_{ij}$ and $m_{ji}$ for distinct blocks $i,j$, define the string $b \in \{0,1\}^M$ by
\begin{equation}\label{eq:direction-string}
 b_a = x_i(P_j(a)) \lor x_j(P_i(a)), \qquad a \in [M].
\end{equation}
If $b$ contains a one, its least nonzero index is $\Delta_{ij}$. At that index, the two defining bits determine the direction between two blocks declared dense: if only the first is one, then $i \to j$; if only the second is one, then $j \to i$; if both are one, the edge is directed from the smaller block index to the larger. A query to $b$ uses four queries to $x$, by computing the two bits, adding their OR to a response bit, and reversing the two queries.

Finding the least nonzero position of $b$ is an instance of minimum finding. We use binary search with quantum search to obtain the fixed query count and small error required for the checking circuits below.

For $L \in [M]$ and a string $u \in \{0,1\}^L$, the following procedure searches for the least index $a \in [L]$ with $u_a = 1$. Starting with $[L]$, it repeatedly searches the earlier half of the remaining interval with error at most $n^{-6}$ and verifies every reported index. It keeps that half if a one is found and the later half otherwise. When one index $a$ remains, it queries $u_a$ and returns $a$ if $u_a = 1$, or $\bot$ otherwise. The unitary implementation stores the search outcomes and intervals.

\begin{lemma}\label{lem:prefix-search}
This procedure has a unitary implementation with fixed query count $O(\sqrt{L \ln n})$. Measuring its output register gives $\bot$ with certainty on the all-zero string and the least nonzero index with probability at least $1 - n^{-5}$ on a nonzero string.
\end{lemma}

\begin{proof}
For $0 < \delta < 1 / 2$, the search algorithm of \cite[Theorem~3]{BuhrmanCleveDeWolfZalka} finds a nonzero position among $L$ bits with error at most $\delta$ using $O(\sqrt{L \ln(1 / \delta)})$ queries. Checking a reported position uses one additional query. Querying all bits gives the alternative bound $L$.

The binary search described before \Cref{lem:prefix-search} takes at most $\lceil\log_2 L\rceil$ steps. At step $d = 0,\ldots,\lceil\log_2 L\rceil - 1$, the interval has at most $\lceil L / 2^d\rceil$ indices, so the earlier half has at most $\ell_d = \lceil L / 2^{d + 1}\rceil$ indices. For a sufficiently large absolute constant $c_1$, the search and verification use at most
\[
 q_d = \left\lceil c_1\min\{\ell_d,\sqrt{\ell_d \ln n}\}\right\rceil
\]
queries. The circuit for this level uses exactly $q_d$ queries. Each oracle call is controlled by the interval register and acts as the identity on the span of basis states in which that register encodes a singleton interval. The implementation stores the verified results and updated intervals in separate registers. Including the final query, the count is
\[
 1 + \sum_{d = 0}^{\lceil\log_2 L\rceil - 1}q_d = O(\sqrt{L \ln n}).
\]
For a nonzero string, if every search is correct, each stored interval contains its least nonzero position. A union bound over $O(\ln n)$ searches of error at most $n^{-6}$ gives error at most $n^{-5}$ for sufficiently large $n$. For the all-zero string, verification prevents a reported nonzero position, and the final query gives zero.
\end{proof}

Put $T(L) = \lceil c_0 \sqrt{L \ln n}\rceil$ for $L \in [M]$. By the choice of $c_0$, the search circuit uses at most $T(L)$ queries to $x$ after the simulation above.

The circuits below implement the comparison of $m_{ij}$ with $m_{ji}$. For distinct $i,j$, set $a_{i,j,L} = 00$ when $L < M$ and $\Delta_{ij} > L$. When $L = M$ or $\Delta_{ij} \leq L$, set $a_{i,j,L} = 10$ if $m_{ji} < m_{ij}$ or if $m_{ji} = m_{ij}$ and $j < i$, and set $a_{i,j,L} = 01$ otherwise. For distinct $i,j$ and $z \in \{0,1\}^2$, define
\begin{equation}\label{eq:checking-unitary}
 \mathcal C_L \ket{i,j,z}
 = \ket{i,j,z \oplus a_{i,j,L}}.
\end{equation}
On basis states with $i = j$, let $\mathcal C_L$ act as the identity.
When both blocks are declared dense, the answers $10$ and $01$ encode the directions $j \to i$ and $i \to j$ in $\mathcal G$, respectively.

For distinct indices, the circuit $\widetilde{\mathcal C}_L$ applies the circuit of \Cref{lem:prefix-search} to $b_1,\ldots,b_L$. If the search returns an index, it queries the two defining bits and uses answer $01$ if only the first bit is one, answer $10$ if only the second is one, and the index rule if both are one; it uses $00$ if both are zero. If the search returns $\bot$, the circuit uses answer $00$ for $L < M$ and the index rule for $L = M$. It adds the answer to the response register and reverses the computation. When the indices are equal, it acts as the identity.

\begin{lemma}\label{lem:edge-implementation}
For every $L \in [M]$, the circuit $\widetilde{\mathcal C}_L$ uses a fixed number of queries, at most $2(T(L) + 2)$. It preserves the indices and satisfies $\widetilde{\mathcal C}_L^2 = I$. If $\Pi_0$ projects its workspace onto zero, then
\begin{equation}\label{eq:edge-implementation-error}
 \left\|(\widetilde{\mathcal C}_L - \mathcal C_L)\Pi_0\right\| \leq n^{-2}.
\end{equation}
\end{lemma}

\begin{proof}
For distinct indices, the two bits queried at the least nonzero position of $b_1,\ldots,b_L$, when present, determine the comparison between $m_{ij}$ and $m_{ji}$, and hence $a_{i,j,L}$. If this prefix is zero and $L < M$, the answer is $00$. If $L = M$ and the whole string is zero, both first indices are $M + 1$, so the index rule determines $a_{i,j,M}$. Thus \Cref{lem:prefix-search} gives the answer in \eqref{eq:checking-unitary} with probability at least $1 - n^{-5}$. When the indices are equal, both $\mathcal C_L$ and $\widetilde{\mathcal C}_L$ act as the identity.

By the definition of $T(L)$, the circuit in \Cref{lem:prefix-search} uses a fixed number of queries, at most $T(L)$. It acts only on the span of basis states with distinct indices and as the identity on the subspace with equal indices. The two subsequent queries act only when the search output register contains an index in $[L]$. The implementation keeps the answer and auxiliary registers, adds the answer to the response register, and reverses the computation. The resulting circuit uses a fixed number of queries, at most $2(T(L) + 2)$. It preserves the indices.

In \Cref{lem:unitary-approximation}, the two indices form the index register, the answer register and auxiliary registers form the workspace, and the register holding $z$ is the response register. The probability bound $n^{-5}$ gives norm error at most $2 n^{-5 / 2} \leq n^{-2}$ for sufficiently large $n$. The same lemma gives $\widetilde{\mathcal C}_L^2 = I$.
\end{proof}

With the classifications available, equality and directions involving a declared sparse block require no queries. For distinct declared dense blocks $v,w$, put $L_\ell = \min\{2^\ell,M\}$ for $0 \leq \ell \leq \lceil \log_2 M\rceil$. The comparison procedure applies $\mathcal C_{L_\ell}$ in increasing order of $\ell$ until a direction is returned. The final prefix has length $M$ and always determines the direction. Let $\ell_*$ be the least index for which $\mathcal C_{L_{\ell_*}}$ returns a direction. If $\Delta_{vw} \leq M$, then $L_{\ell_*} < 2 \Delta_{vw}$; otherwise $L_{\ell_*} = M$. The total query count for stages $0,\ldots,\ell_*$ is at most
\begin{equation}\label{eq:comparison-query-sum}
 \sum_{\ell = 0}^{\ell_*} 2(T(L_\ell) + 2)
 = O\left(\sqrt{L_{\ell_*} \ln n}\right) = O(\theta_{vw}).
\end{equation}

\begin{proof}[Proof of \eqref{eq:edge-square-sum}]
Symmetry follows from $\Delta_{ij} = \Delta_{ji}$. Since $\Delta_{ij} \leq M + 1$, every $\theta_{ij}$ is $O(\sqrt{n \ln n})$. Assume that the test succeeds. If block $i$ is declared sparse, then $\theta_{ij} = 0$ for every $j$. If block $i$ is declared dense, then $|S_i| > n / 3$, so \Cref{lem:first-hit-sum} gives $\sum_{j \neq i}m_{ij} \leq 8 n$. In particular, $\Delta_{ij} \leq m_{ij} \leq M$. By the definition of $\theta_{ij}$,
\[
 \sum_{j = 1}^n \theta_{ij}^2
 \leq (c_0 + 1)^2 \ln n \sum_{\substack{j \neq i\\ \gamma_j = 1}}\Delta_{ij}
 \leq 8(c_0 + 1)^2 n \ln n.
\]
Taking $K = 8(c_0 + 1)^2$ proves \eqref{eq:edge-square-sum}.
\end{proof}

\subsection{Membership checking circuits}\label{subsec:membership}\label{app:membership-implementation}

Fix the input and a table of positions chosen for the test in \Cref{subsec:block-graph}.

For $W = (w_1,\ldots,w_h)$, one has $v \in U(W)$ if and only if $v \notin W$ and $v \to w_d$ for every $d \in [h]$. The procedure tests these conditions in the stored order, using the prefix comparisons from \Cref{sec:edge-proof} for pairs declared dense and the index and classification rules for the other pairs.

The list register consists of a length field and $n$ pairs, each containing a vertex field and a classification bit. Write its computational basis states as $\ket{h} \ket{w_1,g_1}\cdots \ket{w_n,g_n}$. Let $\Pi_{\mathrm{list}}$ project onto the span of these basis states with $h \in \{0,\ldots,n\}$, distinct vertices $w_1,\ldots,w_h \in [n]$, and $g_d = \gamma_{w_d}$ for every $d \in [h]$. Such a basis state encodes $W = (w_1,\ldots,w_h)$; the pairs indexed by $d > h$ are ignored.

As each vertex is appended in \Cref{alg:sourcefinding}, its classification is computed and stored, so the list register remains in this range. The membership computation therefore reads the stored bits $g_d$.

For every integer $q \geq 0$, every basis state in the range of $\Pi_{\mathrm{list}}$ encoding $W = (w_1,\ldots,w_h)$, and every tested vertex $v \in [n]$, let $b_{W,v,q}$ encode $*$ when $q < t_v(W)$, one when $q \geq t_v(W)$ and $v \in U(W)$, and zero otherwise. Define the checking circuit $\mathcal M_q$ for membership in $U(W)$ by
\[
 \mathcal M_q \ket{W,v,z} = \ket{W,v,z \oplus b_{W,v,q}}
\]
for $z \in \{0,1\}^2$. On every other basis state of the list register, or when the tested vertex lies outside $[n]$, $\mathcal M_q$ acts as the identity.

\Cref{alg:truncated-membership} computes $b_{W,v,q}$ using the responses defined by $\mathcal C_{L_\ell}$. Its counter $q'$ records the query counts of $\widetilde{\mathcal C}_{L_\ell}$. Before applying a classification or checking circuit, the membership procedure returns $*$ if completing that circuit would make the query count exceed $q$. On an arbitrary computational basis state, the procedure first checks that $h \in \{0,\ldots,n\}$ and that the tested vertex lies in $[n]$. If these checks pass, it checks that $w_d \in [n]$ for every $d \in [h]$. If any check fails, it returns $*$ before making a query.

\begin{algorithm}[ht]
\caption{Membership in $U(W)$ with at most $q$ queries}
\label{alg:truncated-membership}
\begin{algorithmic}[1]
\Procedure{Membership}{$W,v,q$}
  \State \Return $*$ if $q < s_n$
  \State Classify $v$, store $\gamma_v$, and set $q' \gets s_n$
  \For{$d = 1,\ldots,h$}
    \State \Return $0$ if $v = w_d$
    \If{$g_d \gamma_v = 0$}
      \State \Return $0$ if $g_d = 1$ or if $g_d = \gamma_v = 0$ and $w_d < v$
    \Else
      \For{$\ell = 0,\ldots,\lceil \log_2 M \rceil$}
        \State \Return $*$ if $q'$ plus the query count of $\widetilde{\mathcal C}_{L_\ell}$ exceeds $q$
        \State Apply $\mathcal C_{L_\ell}$ to $(w_d,v)$ on a fresh response register initialized to $00$
        \State Increase $q'$ by the query count of $\widetilde{\mathcal C}_{L_\ell}$
        \If{the response is $10$}
          \State Leave the inner loop
        \ElsIf{the response is not $00$ or $\ell = \lceil \log_2 M \rceil$}
          \State \Return $0$
        \EndIf
      \EndFor
    \EndIf
  \EndFor
  \State \Return $1$
\EndProcedure
\end{algorithmic}
\end{algorithm}

Since $L_{\lceil \log_2 M\rceil} = M$, the circuit $\mathcal C_M$ returns $01$ or $10$ whenever the inner loop is entered. Thus, without the restriction to at most $q$ queries, the procedure would return the membership bit of $U(W)$ for every basis state in the range of $\Pi_{\mathrm{list}}$ and every tested vertex $v \in [n]$.

For a basis state in the range of $\Pi_{\mathrm{list}}$ encoding $W$ and a tested vertex $v \in [n]$, the count $t_v(W)$ defined in \Cref{subsec:candidate-sets} is the final value of $q'$ when the procedure returns zero or one. The procedure using at most $q$ queries returns $*$ when $q < t_v(W)$ and returns the membership bit when $q \geq t_v(W)$. In particular, $t_v(\emptyset) = s_n$. Here $q$ bounds the queries used for one membership test.

\begin{proof}[Proof of \Cref{lem:membership-counts}]
With an empty list, the procedure only classifies the tested vertex, giving $t_v(\emptyset) = s_n$. Let $w \in U(W)$, put $W' = (w_1,\ldots,w_h,w)$, and fix $v \in [n]$. If $v \in U(W) \setminus \{w\}$, the tests against $w_1,\ldots,w_h$ accept, and the additional applications used to determine whether $v \to w$ have total query count $O(\theta_{vw})$. If $v \notin U(W)$, both computations return zero at the least index $d \in [h]$ with $v = w_d$ or $w_d \to v$, so $t_v(W') = t_v(W)$. If $v = w$, both computations perform the tests against $w_1,\ldots,w_h$; the computation for $W'$ then detects equality without a query, so $t_w(W') = t_w(W)$. This proves \eqref{eq:membership-completion-update} and the stated cases of equality.
\end{proof}

The unitary circuit $\mathcal V_q$ implements \Cref{alg:truncated-membership}, replacing each $\mathcal C_{L_\ell}$ by $\widetilde{\mathcal C}_{L_\ell}$ on fresh response and workspace registers initialized to zero. Following \Cref{app:implementations}, it stores the classification, responses, and return values, and copies the first returned value to the internal two-qubit answer register. The circuit makes exactly $q$ queries and preserves the list register and the register holding the tested vertex.

Let $\mathcal H_{\mathrm{list}}$ and $\mathcal H_{\mathrm{vertex}}$ be the state spaces of the list register and the register holding the tested vertex. Let $\mathcal H_{\mathrm{answer}}$ be the state space of the internal answer register, and let $\mathcal H_{\mathrm{aux}}$ be the joint state space of the classification register, the return registers, and the response registers and workspaces used by $\widetilde{\mathcal C}_{L_\ell}$. The separate response register has state space $\mathcal H_{\mathrm{response}}$. Put $\mathcal H_{\mathrm{work}} = \mathcal H_{\mathrm{answer}} \otimes \mathcal H_{\mathrm{aux}}$. The joint state space is
\[
 \mathcal H_{\mathrm{list}} \otimes \mathcal H_{\mathrm{vertex}} \otimes \mathcal H_{\mathrm{work}} \otimes \mathcal H_{\mathrm{response}}.
\]
Let
\[
 \Pi_{\mathrm{work}} = I_{\mathrm{list}} \otimes I_{\mathrm{vertex}} \otimes \ket{0}_{\mathrm{work}} \bra{0}_{\mathrm{work}} \otimes I_{\mathrm{response}}.
\]
A basis state in the range of $\Pi_{\mathrm{work}} \Pi_{\mathrm{list}}$ has the form $\ket{W,v} \ket{0}_{\mathrm{work}} \ket{z}_{\mathrm{response}}$. Extend $\mathcal V_q$ by the identity on $\mathcal H_{\mathrm{response}}$ and $\mathcal M_q$ by the identity on $\mathcal H_{\mathrm{work}}$. Let $\mathsf C \ket{a,z} = \ket{a,z \oplus a}$ add the internal two-qubit answer $a$ to the separate response $z$ and act as the identity on the list, vertex, and auxiliary registers.

\begin{lemma}\label{lem:membership-implementation}
For $q = 0$, put $\widetilde{\mathcal M}_0 = I$; for $q \geq 1$, define
\[
 \widetilde{\mathcal M}_q = \mathcal V_q^\dagger \mathsf C \mathcal V_q.
\]
For every integer $q \geq 0$, the circuit $\widetilde{\mathcal M}_q$ uses $2 q$ queries and satisfies $\mathcal M_q^2 = \widetilde{\mathcal M}_q^2 = I$. Its approximation error satisfies
\begin{equation}\label{eq:membership-error}
 \left\|(\widetilde{\mathcal M}_q - \mathcal M_q)\Pi_{\mathrm{work}} \Pi_{\mathrm{list}}\right\|
 \leq 2 \sqrt{q}\,n^{-2}.
\end{equation}
\end{lemma}

\begin{proof}
For $q = 0$, the inequality $t_v(W) \geq s_n > 0$ gives $b_{W,v,0} = 00$ on valid list states and vertex values. Together with the identity action elsewhere, this gives $\mathcal M_0 = \widetilde{\mathcal M}_0 = I$. Assume that $q \geq 1$.

The circuit $\mathcal V_q$ stores $\gamma_v$ and keeps the response and workspace of each call to $\widetilde{\mathcal C}_{L_\ell}$. Later gates use these responses as controls and preserve them. It starts a classification or checking circuit only if completing that circuit keeps the total query count at most $q$. The convention in \Cref{app:implementations} then gives $\mathcal V_q$ exactly $q$ queries.

Fix a list basis state in the range of $\Pi_{\mathrm{list}}$ and a tested vertex $v \in [n]$. Consider measuring the response register after each application of $\widetilde{\mathcal C}_{L_\ell}$. A response is correct when it agrees with $\mathcal C_{L_\ell}$ for the same indices. Assume that all earlier responses are correct. The next application then has fixed indices and a fresh zero workspace. By \eqref{eq:edge-implementation-error}, its probability of giving an incorrect response is at most $n^{-4}$. At most $q$ such applications occur, since each uses at least one query. If all responses are correct, the internal answer is $b_{W,v,q}$ by the definition of $t_v(W)$. A union bound on the first incorrect response bounds the probability of a different internal answer by $q n^{-4}$. The same bound holds when these measurements are deferred.

In the tensor decomposition preceding \Cref{lem:membership-implementation}, $\mathcal H_{\mathrm{list}} \otimes \mathcal H_{\mathrm{vertex}}$ is the index space of \Cref{lem:unitary-approximation}. Both registers are preserved by $\mathcal V_q$. Applying that lemma with $\delta = q n^{-4}$ gives error at most $2 \sqrt{q}\,n^{-2}$ on the valid list and vertex subspace with zero workspace. If the tested vertex lies outside $[n]$, the initial range check leaves the internal answer at $00$, so $\widetilde{\mathcal M}_q$ and $\mathcal M_q$ agree on zero workspace. Preservation of the register holding the tested vertex combines these estimates into \eqref{eq:membership-error}.

The definition of the action on the response register gives $\mathcal M_q^2 = I$, and \Cref{lem:unitary-approximation} gives $\widetilde{\mathcal M}_q^2 = I$. The circuit $\mathsf C$ uses no queries, while $\mathcal V_q$ and its inverse each use $q$, giving the claimed count $2 q$.
\end{proof}

Fix a list basis state in the range of $\Pi_{\mathrm{list}}$ encoding $W$. On this subspace, the checking circuits $\mathcal M_q$ have the form assumed in \Cref{app:variable-search}, with $G = U(W)$ and $t_v = t_v(W)$. By \Cref{lem:membership-implementation}, each circuit $\widetilde{\mathcal M}_q$ uses $2 q$ queries.

\subsection{Estimates for candidate sets}\label{app:candidate-estimates}

\begin{definition}\label{def:candidate-process}
For a directed complete graph on $[n]$, an absolute constant $C \geq 1$, and an initial set $U_0 \subseteq [n]$, the process starts with $W_0 = \emptyset$. Given $W_h = (w_0,\ldots,w_{h - 1})$ and $U_h$, step $h$ proceeds as follows.
\begin{enumerate}
\item If $U_h = \emptyset$, the process stops.
\item If $U_h \neq \emptyset$, each vertex has an assigned selection probability at most $C / |U_h|$, and these probabilities sum to at most one. They may depend on $W_h$.
\item The process selects $w_h = v$ with the probability assigned to $v$, for each $v \in U_h$, and stops with the remaining probability.
\item If $w_h$ is selected, put $W_{h + 1} = (w_0,\ldots,w_h)$ and $U_{h + 1} = \{v \in U_h : v \to w_h\}$.
\end{enumerate}
After the process stops, all later sets are $\emptyset$. Let $W_{\mathrm{end}}$ be the final ordered list of selected vertices.
\end{definition}

\begin{proof}[Proof of \Cref{lem:candidate-estimates}]
We first bound the expected decrease of $m_h$ to prove the first two estimates in \Cref{lem:candidate-estimates}. We then use this decrease to bound $\mathbb E \sum_v t_v(W_{\mathrm{end}})^2$.

Fix the choices preceding a step with $m_h = m > 0$. If the process stops, $m_{h + 1} = 0$. Otherwise, $m - m_{h + 1}$ is one plus the outdegree of the selected vertex in $U_h$. For $0 \leq b < m$, at most $2b$ choices of vertex give a decrease of at most $b$. Indeed, if there are $d > 0$ such vertices, each has outdegree at most $b - 1$ within $U_h$, so $\binom{d}{2} \leq d(b - 1)$ and $d \leq 2b - 1$. Since each vertex has selection probability at most $C / m$, the probability over the outcome of this step satisfies $\Pr[m - m_{h + 1} \leq b] \leq 2Cb / m$. The identity $\mathbb E[m - m_{h + 1}] = \int_0^m \Pr[m - m_{h + 1} > b]\,db$ therefore gives
\begin{equation}\label{eq:candidate-decrease}
 \mathbb E[m - m_{h + 1}] \geq \int_0^{m / (2C)} \left(1 - \frac{2Cb}{m}\right)\,db = \frac{m}{4C}.
\end{equation}

For every $a > 0$, averaging \eqref{eq:candidate-decrease} over the preceding choices gives
\[
 \mathbb E \sum_{h:m_h > 0} m_h^{1 - a}
 \leq 4C\,\mathbb E \sum_{h:m_h > 0} \frac{m_h - m_{h + 1}}{m_h^a}
 \leq 4C \sum_{j = 1}^n j^{-a}.
\]
The last inequality holds because the integer intervals $\{m_{h + 1} + 1,\ldots,m_h\}$ are disjoint and $j^{-a} \geq m_h^{-a}$ on each such interval. Taking $a = 1$ and $a = 3 / 2$ gives $\mathbb E R = O(\ln n)$ and $\mathbb E \sum_{h:m_h > 0} m_h^{-1 / 2} = O(1)$.

For the last estimate in \eqref{eq:completion-profile}, we bound the successive increases in $t_v(W_h)$ while $v \in U_h$. By induction using \eqref{eq:candidate-update} and $U_0 = U(\emptyset) = [n]$, one has $U_h = U(W_h)$ whenever step $h$ is reached. Thus \eqref{eq:membership-completion-update} applies whenever $w_h$ is selected.

Use a vector $\boldsymbol{c}_h \in \mathbb R^n$ to bound the additional queries at step $h$. If $w_h$ is selected, put $c_h(v) = \theta_{v w_h}$ for $v \in U_h$ and $c_h(v) = 0$ otherwise. If no vertex is selected at step $h$, put $\boldsymbol{c}_h = 0$ and set all later vectors to zero. Whenever $w_h$ is selected, \eqref{eq:membership-completion-update} bounds $t_v(W_{h + 1}) - t_v(W_h)$ by an absolute constant times $c_h(v)$. For $v \notin U_h$, the first test excluding $v$ compares it with an entry of $W_h$ and remains unchanged after the append, so $t_v(W_{h + 1}) = t_v(W_h)$ and $c_h(v) = 0$.

Put $V = \mathbb E\|\sum_h \boldsymbol{c}_h\|^2$. Iterating \eqref{eq:membership-completion-update} shows that $t_v(W_{\mathrm{end}})$ is at most $s_n$ plus an absolute constant times $\sum_h c_h(v)$. Therefore,
\[
 \mathbb E \sum_{v = 1}^n t_v(W_{\mathrm{end}})^2 = O(n s_n^2 + V).
\]
It therefore suffices to prove $V = O(n \ln^2 n)$.

Symmetry of $\theta_{ij}$ and \eqref{eq:edge-square-sum} give $\|\boldsymbol{c}_h\|^2 = \sum_{v \in U_h} \theta_{v w_h}^2 = \sum_{v \in U_h} \theta_{w_h v}^2 = O(n \ln n)$ whenever a vertex is selected. Hence $\sum_h \|\boldsymbol{c}_h\|^2 = O(n \ln n\,R)$.

To bound the cross terms in $V$, put $\lambda = 1 - 1 / (32C)$. Given the preceding choices, Jensen's inequality and \eqref{eq:candidate-decrease} give
\begin{equation}\label{eq:candidate-contraction}
 \mathbb E \sqrt{m_{h + 1}} \leq \sqrt{1 - \frac{1}{4C}} \sqrt{m_h} \leq \lambda \sqrt{m_h}.
\end{equation}
The same estimate holds when $m_h = 0$, since all subsequent sets are empty.

When $m_h > 0$, let $\boldsymbol{\chi}_h \in \mathbb R^n$ be the normalized characteristic vector of $U_h$, with coordinates $\chi_h(v) = 1 / \sqrt{m_h}$ for $v \in U_h$ and $\chi_h(v) = 0$ otherwise. When $m_h = 0$, put $\boldsymbol{\chi}_h = 0$.

For $h' \geq h$ and $m_h > 0$, nesting gives $\langle \boldsymbol{\chi}_h,\boldsymbol{\chi}_{h'}\rangle = \sqrt{m_{h'} / m_h}$. After fixing the choices preceding step $h$, iterate \eqref{eq:candidate-contraction}; averaging then gives $\mathbb E\langle \boldsymbol{\chi}_h,\boldsymbol{\chi}_{h'}\rangle \leq \lambda^{h' - h} \Pr[m_h > 0]$. Summing the diagonal and cross terms yields
\begin{equation}\label{eq:occupation-vector}
 \mathbb E\left\|\sum_h \boldsymbol{\chi}_h\right\|^2
 \leq \left(1 + 2 \sum_{j \geq 1} \lambda^j\right) \mathbb E R
 = O(\ln n).
\end{equation}

At a reached step $h$ with $m_h > 0$, \Cref{def:candidate-process} assigns selection probability at most $C / m_h$ to each $w \in U_h$. For $v \in U_h$, Cauchy-Schwarz and \eqref{eq:edge-square-sum} therefore give
\begin{equation}\label{eq:completion-conditional-mean}
 0 \leq \mathbb E[c_h(v) \mid W_h] \leq \frac{C}{m_h} \sum_{w \in U_h} \theta_{vw} = O\left(\sqrt{n \ln n}\,\chi_h(v)\right).
\end{equation}
For $v \notin U_h$, both $c_h(v)$ and $\chi_h(v)$ are zero.

For $h < j$, if step $j$ is reached, the vectors $\boldsymbol{c}_h$ and $\boldsymbol{\chi}_j$ are determined by the preceding choices. Applying \eqref{eq:completion-conditional-mean} coordinatewise after fixing those choices and then averaging gives
\begin{equation}\label{eq:completion-cross-term}
 \mathbb E\langle\boldsymbol{c}_h,\boldsymbol{c}_j\rangle
 = O(\sqrt{n \ln n})\,\mathbb E\langle\boldsymbol{c}_h,\boldsymbol{\chi}_j\rangle.
\end{equation}
The later vectors are zero after stopping, so this estimate holds for every pair $h < j$. Expanding the squared norm, summing \eqref{eq:completion-cross-term}, and using nonnegativity gives
\[
 V = O\left(n \ln n\,\mathbb E R + \sqrt{n \ln n}\,\mathbb E\left\langle \sum_h \boldsymbol{c}_h,\sum_h \boldsymbol{\chi}_h\right\rangle\right).
\]
The bound $\mathbb E R = O(\ln n)$, Cauchy-Schwarz, and \eqref{eq:occupation-vector} bound the RHS by $O(n \ln^2 n) + O\left(\sqrt{n \ln^2 n\,V}\right)$, which implies $V = O(n \ln^2 n)$. Together with $s_n = O(\ln n)$, this bound proves the last estimate in \eqref{eq:completion-profile}.
\end{proof}

\subsection{Implementation of the source algorithm}\label{app:source-implementation}

In the circuit for \textsc{FindSource} using $\mathcal M_q$, each possible return in \Cref{alg:sourcefinding} is written to a fresh register initialized to $\ket{\mathrm{init}}$. Later stages, including their oracle calls, are applied only while all earlier return registers remain in this state. Each search circuit and its inverse are expanded as a sequence of applications of $\mathcal M_q$ or its inverse interleaved with input-independent unitary gates. Assign a separate workspace $\mathcal H_{\mathrm{work},a}$, initialized to $\ket{0}_{\mathrm{work},a}$, to each occurrence $a$. The implemented circuit replaces each occurrence by $\widetilde{\mathcal M}_q$ or its inverse. Every other gate, including each search reflection, is extended by $\bigotimes_a I_{\mathrm{work},a}$.

For unitaries $\mathcal U_1,\ldots,\mathcal U_s$ and replacements $\widetilde{\mathcal U}_1,\ldots,\widetilde{\mathcal U}_s$, the identity
\begin{equation}\label{eq:unitary-telescoping}
 \widetilde{\mathcal U}_s\cdots\widetilde{\mathcal U}_1 - \mathcal U_s\cdots \mathcal U_1
 = \sum_{a = 1}^s\widetilde{\mathcal U}_s\cdots\widetilde{\mathcal U}_{a + 1}
 (\widetilde{\mathcal U}_a - \mathcal U_a)\mathcal U_{a - 1}\cdots \mathcal U_1
\end{equation}
will bound the accumulated approximation error.

\begin{proof}[Proof of \Cref{lem:source-implementation}]
Compare two unitary circuits, one using $\mathcal M_q$ and one using $\widetilde{\mathcal M}_q$. The guarantees in \Cref{lem:search-stage} apply to the circuit using $\mathcal M_q$. The final states of the two circuits are compared using \eqref{eq:unitary-telescoping}. Fix the table of positions in the test defined in \Cref{subsec:block-graph}. Both circuits use the unitary implementation of \textsc{Search} given in \Cref{subsec:complete-search}. They prepare the initial vertex register in uniform superposition over $[n]$ and give every call to \textsc{Search} a separate output register, whose computational basis value controls the later operations without being measured. The classification $\gamma_{w_d}$ of each stored vertex $w_d$ is written to the field $g_d$. Later gates preserve each output register and classification field after it has been written. Their joint computational basis values therefore index orthogonal subspaces preserved by the later gates, and measuring the final output register gives the output distribution of the procedure using $\mathcal M_q$ for the fixed table.

By \eqref{eq:search-circuit-count}, each complete search contains $O(\sqrt{n})$ membership checking circuits or their inverses. The loop in \Cref{alg:sourcefinding} makes at most $n - 1$ searches, so the circuit for \textsc{FindSource} contains $O(n^{3 / 2})$ such occurrences, including those controlled by earlier return values.

Before an occurrence is performed in the circuit using $\mathcal M_q$, its assigned workspace is zero and the list register lies in the range of $\Pi_{\mathrm{list}}$. The initialization and append operations store classifications computed from the fixed table, while every search circuit and each reflection in it leaves the list register unchanged. Both circuits act as the identity whenever the stored control values prevent that occurrence from being applied. These controls are unchanged by the membership checking circuits. Since every parameter satisfies $q \leq q_J < 12 T_2$, \Cref{lem:membership-implementation} and orthogonality of the control subspaces bound the error at each occurrence by $2 \sqrt{12 T_2} n^{-2}$ on every state reached by the preceding circuit using $\mathcal M_q$. The same bound applies to inverse occurrences because both membership checking circuits square to the identity.

Apply \eqref{eq:unitary-telescoping} to all occurrences, replacing them from last to first. The intervening gates are common to both circuits and unitary, so the distance between the final states is bounded by
\begin{equation}\label{eq:source-implementation-error}
 O\left(n^{3 / 2} \sqrt{T_2} n^{-2}\right) = O\left(n^{-1 / 4} \sqrt{\ln n}\right) = o(1).
\end{equation}
The estimate in \Cref{lem:membership-implementation} holds for every fixed table of positions. Because the table register is preserved, distinct tables determine orthogonal subspaces. The uniform error bound therefore also holds when the table register is in superposition. For the measurement of the final output register, each outcome probability differs between the two circuits by at most twice the distance between their final states. Thus each difference is $o(1)$, uniformly over inputs and tables.

Before applying any circuit that computes a classification or membership answer, or the inverse of such a circuit, the procedure returns $1$ if the application would make the total query count exceed $T_{\max}$. Consequently, the procedure makes at most $T_{\max}$ queries on every input and for every table of positions. The construction described in \Cref{app:implementations} gives a circuit making exactly $T_{\max}$ queries. Measuring its output register gives the distribution of the implemented procedure.

The random table is prepared in an unchanged register in uniform superposition over $[n]^{n \times s_n}$. Its basis values determine the queried positions throughout the circuit. Orthogonality of the basis states of the table register gives the average of the output distributions for the fixed tables, with the same query count $T_{\max}$.
\end{proof}

\subsection{Exact search}\label{app:exact-preparation}

\begin{proof}[Proof of \Cref{lem:weighted-success}]
On a negative input, verification rejects and $\mu(\omega) = 0$. Fix a positive input with satisfied term $(i^\star,S)$. With the choices before round $a$ fixed and $i^\star \in U_a$, the label paired with $i^\star$ is uniform among the other candidates, and either order in the pair is equally likely. A comparison with block $j$ keeps $i^\star$ if it reaches $m_j(S)$. Otherwise both prefixes contain only zeros, and the comparison deletes $i^\star$ with probability $1 / 2$. Comparisons with auxiliary labels always keep $i^\star$. The conditional probability $p_a$ of keeping $i^\star$ is therefore given by \eqref{eq:elimination-retention-probability}.

On an accepted record, $p_a = p_a(\omega)$. Since $1 / (2(m - 1)) \leq 1 / m$ for $m \geq 2$ and $|U_a| = N / 2^a$,
\begin{equation}\label{eq:elimination-deletion-sum}
 \sum_{a = 0}^{L - 1}(1 - p_a)
 \leq \frac{1}{N}\sum_{j \neq i^\star}\sum_{\substack{a \geq 0\\32 \cdot 2^a < m_j(S)}}2^a
 \leq \frac{1}{16 N}\sum_{j \neq i^\star}m_j(S)
 \leq \frac{1}{2}.
\end{equation}
The inner sum is less than $m_j(S) / 16$, and the last bound uses \Cref{lem:first-hit-sum}. Hence $\prod_a p_a \geq 1 - \sum_a(1 - p_a) \geq 1 / 2$, which gives $\mu(\omega) \leq 2$. This estimate holds for every $k$-element set $S$ and all sets $U_a$ of sizes $N / 2^a$ containing $i^\star$.

For the expectation, consider the product of the reciprocal probabilities from a given round onward, with value zero if $i^\star$ is subsequently deleted. We prove by backward induction that its conditional expectation is one whenever the preceding choices have kept $i^\star$. At the final singleton $\{i^\star\}$, verification accepts and the empty product is one. At round $a$, the choices that keep $i^\star$ have total probability $p_a$; after each such choice the expected product of the later reciprocals is one by induction. Multiplication by $p_a^{-1}$ therefore gives expectation one. Since $U_0 = [N]$ contains $i^\star$, this proves \eqref{eq:weighted-success}.
\end{proof}

We construct the circuit $\mathcal A_x$ used by \textsc{FindTerm} (\Cref{alg:exact-term}). Put $n_a = N / 2^a$ for $0 \leq a \leq L$. The record registers are
\[
\begin{array}{c|l}
 \text{register} & \text{contents} \\ \hline
 \pi_a,\quad 0 \leq a < L & \text{a permutation of }[n_a] \\
 U_a,\quad 0 \leq a \leq L & \text{an }N\text{-bit characteristic vector} \\
 R_a,\quad 0 \leq a < L & 32 N\text{ answer bits from round }a \\
 R_L & n + M\text{ verification answer bits} \\
 c & \text{the verification bit}
\end{array}
\]
Let $\mathcal H_{\mathrm{record}}$ be their joint state space. The flag qubit has state space $\mathbb C^2$, and the query address and auxiliary registers have joint state space $\mathcal H_{\mathrm{aux}}$. The circuit acts on
\[
 \mathcal H_{\mathrm{record}} \otimes \mathbb C^2 \otimes \mathcal H_{\mathrm{aux}},
\]
with every register initially zero.

Each permutation has a distinct nonzero binary code in $\lceil\log_2(n_a! + 1)\rceil$ bits; zero denotes initialization. The query address register has $\lceil\log_2(n^2 + 1)\rceil$ bits. Writing a classical function $g(u)$ in a target register means applying $\ket{u,v} \mapsto \ket{u,v \oplus g(u)}$, leaving its arguments unchanged. Unused permutation codes and candidate sets of the wrong cardinality give computed value zero and disable the corresponding reads. Each read computes a legal address from the stored indices, adds the addressed bit to the next answer register, and undoes the address computation. A disabled call acts as the identity by the convention in \Cref{app:implementations}. The address and auxiliary registers are restored to zero after each call.

The circuit performs the following steps.
\begin{enumerate}
\item It writes $[N]$ in $U_0$. For every $a < L$, it exchanges the zero state of $\pi_a$ with
\[
 \frac{1}{\sqrt{n_a!}}\sum_{\pi \in \mathfrak S_{n_a}}\ket{\pi},
\]
where $\mathfrak S_{n_a}$ is the set of permutations of $[n_a]$, and fixes the orthogonal complement of these two states. They are orthogonal because the permutation codes are nonzero.

\item For $a = 0,\ldots,L - 1$, the permutation $\pi_a$ acts on the ranks of the sorted elements of $U_a$, and consecutive entries form ordered pairs. For each pair $(i,j)$, the circuit reads $x_i(P_j(t))$ and then $x_j(P_i(t))$ for $t = 1,\ldots,32 \cdot 2^a$, storing the answers consecutively in $R_a$. When either label is auxiliary or $t > M$, both calls act as the identity. Thus every pair has the same number of calls, and each round uses $32 N$ queries. The comparison rule in \Cref{alg:elimination} determines the retained labels, whose characteristic vector is written in $U_{a + 1}$.

\item If $U_L$ is a singleton $\{i\}$ with $i \in [n]$, the circuit reads $x_i(1),\ldots,x_i(n)$ into the first $n$ bits of $R_L$. Let $S$ be their support. When $|S| = k$, it queries $x_j(P_i(a))$ for $j \neq i$ and $a \in [m_j(S)]$, in increasing order of $j$ and then $a$, into the remaining $M$ bits. The bound in \Cref{lem:first-hit-sum} ensures that they fit. Calls beyond this enumeration act as the identity, as do all $M$ calls when $|S| \neq k$. If $U_L$ is not such a singleton, all $n + M$ calls act as the identity. The bit $c$ is one if $U_L = \{i\}$ with $i \in [n]$, $|S| = k$, and the last $M$ bits of $R_L$ are zero; otherwise it is zero.

\item Let $\omega$ denote the basis record in $\mathcal H_{\mathrm{record}}$. If $c = 1$, $U_L = \{i^\star\} \subseteq [n]$, the support $S$ of the first $n$ bits of $R_L$ has size $k$, and every $U_a$ has size $n_a$ and contains $i^\star$, define $p_a(\omega)$ by \eqref{eq:elimination-retention-probability} with $i = i^\star$, and define $\mu(\omega)$ by \eqref{eq:elimination-weight}. Set $\mu(\omega) = 0$ otherwise. The circuit applies
\[
 \sum_\omega \ket{\omega}\bra{\omega} \otimes
 \begin{pmatrix}
  \sqrt{1 - \mu(\omega) / 4} & -\sqrt{\mu(\omega) / 4} \\
  \sqrt{\mu(\omega) / 4} & \sqrt{1 - \mu(\omega) / 4}
 \end{pmatrix}
 \otimes I_{\mathrm{aux}}.
\]
The estimate \eqref{eq:elimination-deletion-sum} bounds the assigned weights by two on every basis record. This defines an input-independent unitary on the full state space and gives \eqref{eq:exact-flag-rotation} when the flag qubit is initialized to zero.
\end{enumerate}

The rounds and verification use $32 N L + n + M = T_n$ queries; the other steps use none. Reversing the gates gives $\mathcal A_x^{-1}$ with the same count. Before the flag rotation, the amplitudes are products of the uniform permutation amplitudes, and the later gates preserve distinct records. Their probabilities are therefore those of \Cref{alg:elimination}, which proves \eqref{eq:exact-flag-probability}.

\section{Construction and totalization}\label{app:construction}

\subsection{Construction and certificate complexity}\label{app:dnf-partial}\label{app:description-verification}\label{app:certificate-lower}

\begin{proof}[Proof of \Cref{lem:first-hit-sum}]
Fix $S$ and put $p = |S|/n \geq 1/3$. Extend each list by independent uniform entries. The first index of an entry in $S$ in each extension has a geometric distribution with parameter $p$ and is at least $m_j(S)$. Hence
\[
 \mathbb E[(5/4)^{m_j(S)}]
 \leq \sum_{a \geq 1} p(1-p)^{a-1}(5/4)^a
 = \frac{5p}{5p-1} \leq \frac{5}{2}.
\]
The indices for different lists are independent. Markov's inequality and a union bound over at most $2^n$ sets bound the probability that $\sum_{j = 1}^n m_j(S) > 8n$ for some $S$ with $|S| \geq n / 3$ by
\[
 2^n(5/2)^n(4/5)^{8n} = [5(4/5)^8]^n = o(1),
\]
since $5(4/5)^8 < 1$.
\end{proof}

\begin{proof}[Proof of \Cref{lem:undefined-input}]
Following \cite[Section~5]{pabbaraju2026}, let $\eta$ be a partial assignment of size below $\tau_n$ consistent with the all-zero input. Fewer than $\tau_n = n(n - k + 1)$ pairs $(i,p)$ have an $x$- or $y$-bit fixed by $\eta$. Hence some block $i$ has at most $n - k$ positions at which either bit is fixed, leaving at least $k$ positions at which both bits are unfixed. Choose $k$ of these positions to form $S$.

Set $x_i(p) = 1$ for $p \in S$ and every other unfixed bit of $x$ to zero. Every bit of $x$ fixed by $\eta$ is zero, so every bit required by the term to be zero has value zero in this completion; hence the completion satisfies $(i,S)$. Setting every unfixed bit of $y$ to zero gives an input in the domain of $H$ with value zero. For an input in the domain of $H$ with value one, choose any $p \in S$, set $y_i(p) = 1$, and set every other unfixed bit of $y$ to zero.
\end{proof}

We define the description $z$ and the verification procedure $A$ by adapting \cite[Section~5]{pabbaraju2026}. The string $z$ encodes the tuple
\[
 (i,w_1,\ldots,w_n,\sigma_1,\ldots,\sigma_n,j_1,\ldots,j_M,p^\star),
\]
where $i \in [n]$, every $w_p,j_t,p^\star \in \{0,\ldots,n\}$, and every $\sigma_j \in \{0,\ldots,M\}$. The fields $i,w_p,j_t,p^\star$ each use $\lceil \log_2(n + 1) \rceil$ bits, and each $\sigma_j$ uses $\lceil \log_2(M + 1) \rceil$ bits. Each field uses the binary encoding of its value. Thus $d = \Theta(n \ln n)$.

The value $A(x,y,\beta,z)$ is one if and only if all field values lie in these ranges and the following checks accept:
\begin{enumerate}
\item \textbf{Count checks.} With $w_0 = 0$, these checks require $w_n = k$ and $w_p = w_{p - 1} + x_i(p)$ for every $p \in [n]$.

\item \textbf{Length checks.} Put $\sigma_0 = 0$ and $\ell_j = \sigma_j - \sigma_{j - 1}$ for $j \in [n]$. The check for $j$ requires $\ell_j = 0$ when $j = i$, and requires $1 \leq \ell_j \leq M$ and $x_i(P_j(\ell_j)) = 1$ when $j \neq i$. The field $\sigma_j$ specifies the total prefix length for blocks in $[j]$, omitting block $i$.

\item \textbf{Zero checks.} For each $t \in [M]$ with $t > \sigma_n$, the check requires $j_t = 0$. Otherwise it takes $j = j_t$, requires $j \in [n]$ and $\sigma_{j - 1} < t \leq \sigma_j$, puts $a = t - \sigma_{j - 1}$, and checks $x_j(P_i(a)) = 0$. For $t \leq \sigma_n$, the field $j_t$ specifies the block checked at index $t$, and the cumulative lengths determine the position $a$ in its prefix.

\item \textbf{Value checks.} For $\beta = 0$, the checks require $p^\star = 0$ and, for every $p \in [n]$, require $y_i(p) = 0$ whenever $x_i(p) = 1$. For $\beta = 1$, they require $p^\star \in [n]$ and $x_i(p^\star) = y_i(p^\star) = 1$.
\end{enumerate}

There is one range check for each field. Before using an encoded or derived value as an index, each check verifies that it lies in the set being indexed and rejects otherwise.

\begin{proof}[Proof of \Cref{lem:description-verification}]
Suppose a description is accepted and put $S = \{p : x_i(p) = 1\}$.  The count checks give $w_p = |S \cap [p]|$ by induction and hence $|S| = k$.  For $j \neq i$, the length check gives $P_j(\ell_j) \in S$ and therefore $m_j(S) \leq \ell_j$.

The intervals $(\sigma_{j - 1},\sigma_j]$ partition $1,\ldots,\sigma_n$, with an empty interval only at $j = i$.  For $1 \leq a \leq \ell_j$, the check at $t = \sigma_{j - 1} + a$ therefore forces $j_t = j$ and verifies $x_j(P_i(a)) = 0$.  Since $m_j(S) \leq \ell_j$, these checks include every position in $A_{ij}(S)$.  Thus $(i,S)$ is satisfied, and the value checks establish $H(x,y) = \beta$.

Conversely, suppose $x$ satisfies $(i,S)$. Set $w_p = |S \cap [p]|$, take $\ell_i = 0$ and $\ell_j = m_j(S)$ for $j \neq i$, and put $\sigma_j = \sum_{a = 1}^j \ell_a$. By \Cref{lem:first-hit-sum}, $\sigma_n \leq \sum_{j = 1}^n m_j(S) \leq M$. Set $j_t = j$ when $\sigma_{j - 1} < t \leq \sigma_j$, and set $j_t = 0$ for $t > \sigma_n$. For $j \neq i$ and $1 \leq a \leq \ell_j$, the position $P_i(a)$ belongs to $A_{ij}(S)$, so $x_j(P_i(a)) = 0$ and every zero check accepts. Take $p^\star = 0$ when $H(x,y) = 0$, and otherwise take any $p^\star \in S$ with $y_i(p^\star) = 1$.

The count, length, and value checks are indexed over $[n]$, and the zero checks over $[M]$. Together with one range check for each field, this gives $O(n)$ checks with indices independent of encoded values.  Each check reads $O(1)$ fields of $O(\ln n)$ bits and at most two bits of $(x,y)$, so it can be evaluated deterministically with $O(\ln n)$ queries.
\end{proof}

\begin{proof}[Proof of \Cref{lem:certificate-lower}]
The proof follows the all-zero argument of \cite[Section~5]{pabbaraju2026}.  At $Z_0$, each input to $H$ lies outside the domain of $H$, so $F_n(Z_0) = 0$.  Let $\eta$ be a partial assignment consistent with $Z_0$ and of size below $\tau_n$.  At most $|\eta|$ cells contain a cheat-sheet coordinate fixed by $\eta$, while the array has more than $\tau_n$ cells.  Hence some cell $a$ contains no fixed coordinate.

For every $\ell \in [r]$, apply \Cref{lem:undefined-input} to the restriction of $\eta$ to the $\ell$th input to $H$ and complete it to value $a_\ell$.  Fill cell $a$ with descriptions whose existence is established by \Cref{lem:description-verification}, and set every remaining unfixed bit of the cheat-sheet array to zero.  This completion is consistent with $\eta$ and has $F_n = 1$, so $\eta$ is not a zero-certificate.

For the reverse inequality, choose one of the $r$ inputs to $H$ and fix any $n - k + 1$ bits to zero in every block of $x$ in that input.  Each block then contains at most $k - 1$ ones in every completion, so this input remains outside the domain of $H$.  By \Cref{lem:description-verification}, every cell rejects.  The fixed bits therefore form a zero-certificate for $Z_0$ of size $n(n - k + 1) = \tau_n$.
\end{proof}

\subsection{Reduction to evaluating \texorpdfstring{$H$}{H}}\label{subsec:query-transfer}\label{app:query-transfer}

\begin{proof}[Proof of \Cref{lem:query-transfer}]
On the basis state $\ket{\ell}$ of the selection register, write the action of the controlled circuit for $H$ as
\[
 \ket{\ell} \ket{0}_{\mathrm{out}} \ket{0}_{\mathrm{aux}}
 \longmapsto
 \ket{\ell} \sum_{b \in \{0,1\}} \ket{b}_{\mathrm{out}} \ket{\psi_{\ell,b}}_{\mathrm{aux}}.
\]
The vectors $\ket{\psi_{\ell,b}}$ are subnormalized. If the $\ell$th input lies in the domain of $H$ and has value $\beta$, then $\|\psi_{\ell,\beta}\|^2 \geq 3 / 4$. The selection register determines which stored input is queried and remains unchanged, so this circuit and its inverse each use $T_H$ queries. Apply the recovery algorithm of \cite[Theorem~3]{BuhrmanNewmanRohrigDeWolf} with $\epsilon = 1 / 4$ and $t = r$, and measure its output register to obtain $a \in \{0,1\}^r$. If every input lies in the domain of $H$, then $a$ equals their string of values with probability at least $2 / 3$. On every input, the recovery algorithm applies the controlled circuit for $H$ or its inverse $O(r)$ times and therefore uses $O(r T_H)$ queries.

By \Cref{lem:description-verification}, cell $a$ has $O(n r)$ checks, each evaluated deterministically with $O(\ln n)$ queries. The convention in \Cref{app:implementations} gives these checks a common fixed query count. To detect rejection, the algorithm applies Grover search over the checks. For each check, it computes whether the check rejects, adds this value to a response qubit, and reverses the computation, using $O(\ln n)$ queries to the input and the description. The search is run with error at most $1 / 200$, uses $O(\sqrt{n r} \ln n)$ queries by \cite[Theorem~3]{BuhrmanCleveDeWolfZalka}, and verifies every reported rejection. The algorithm outputs one if and only if the search returns $\bot$.

If $F_n = 0$, every cell rejects, so the search for a rejected check has error probability at most $1 / 200$ for every returned address. If $F_n = 1$, all inputs lie in the domain of $H$, and their values specify the accepting cell. The recovery algorithm returns its address with probability at least $2 / 3$. At that address every check accepts, so the search for a rejected check returns $\bot$ with certainty.
\end{proof}

\subsection{Bijections associated with pairs of blocks}\label{app:bijection-extension}

We prove the bounds for the original construction stated in \Cref{subsec:construction-comparison}. Number the positions within each block by $[n]$, and write $\lambda_{ij}:[n] \to [n]$ for the bijection on block $i$ associated with the pair $\{i,j\}$ in \cite[Section~3]{pabbaraju2026}. For nonempty $S \subseteq [n]$, put $m_{i,j}(S)=\min_{p \in S}\lambda_{ij}(p)$; put $m_{i,j}(\emptyset)=n+1$. The term $(i,S)$ requires block $i$ to have support $S$ and requires
\[
 x_j(\lambda_{ji}^{-1}(a))=0
 \qquad \text{for all } j \neq i,\ 1 \leq a \leq m_{i,j}(S).
\]

Choose all $n(n-1)$ bijections independently and uniformly. For fixed $i,S$ with $|S| \geq n/3$ and $0 \leq t \leq n$,
\[
 \Pr[m_{i,j}(S)>t]
 = \frac{\binom{n-t}{|S|}}{\binom{n}{|S|}}
 \leq (1-|S|/n)^t \leq (2/3)^t.
\]
The calculation in the proof of \Cref{lem:first-hit-sum}, with an additional union bound over $i$, shows that
\begin{equation}\label{eq:bijection-prefix-sum}
 \sum_{j \neq i}m_{i,j}(S) \leq 8n,
 \qquad \text{for all } i \in [n],\ |S| \geq n/3,
\end{equation}
with probability at least $1 - n[5(4/5)^8]^n = 1 - o(1)$.
Fix such a choice. In particular, every term has at most $9n$ literals.

For an arbitrary input $x$, let $S_i$ be the support of block $i$ and put $m_{ij}=m_{i,j}(S_i)$. If $(i,S)$ is satisfied, its zero requirements imply $m_{ji}>m_{ij}$ for every $j \neq i$. The classification test and edge rule from \Cref{subsec:block-graph} therefore preserve this source whenever the test succeeds.

The comparison circuits use the string
\[
 b_a=x_i(\lambda_{ij}^{-1}(a)) \lor x_j(\lambda_{ji}^{-1}(a)),
 \qquad a \in [n].
\]
A query to this string uses four queries to $x$, and, when the string contains a one, its least nonzero index is $\min\{m_{ij},m_{ji}\}$. Substituting these addresses in \Cref{sec:edge-proof} and ending the prefix comparisons at length $n$ gives the same bounds on queries and approximation errors. Define the symmetric parameters $\theta_{ij}$ as in \Cref{subsec:block-graph}. If the classification test succeeds and block $i$ is declared dense, \eqref{eq:bijection-prefix-sum} gives
\[
 \sum_j\theta_{ij}^2
 = O\left(n+\ln n\sum_{j \neq i}m_{i,j}(S_i)\right)
 = O(n \ln n);
\]
for a block declared sparse, the sum is zero. Thus \eqref{eq:edge-square-sum} holds. The proof of \Cref{lem:query-upper} applies with these comparison circuits.

The construction of $G$ in \cite{pabbaraju2026} uses $r = \Theta(\ln n)$ copies of $H$ and $O(n r)$ checks per cell, each evaluated deterministically with $\operatorname{polylog} n$ queries. The proof of \Cref{lem:query-transfer} therefore gives $\operatorname{Q}(G) = \widetilde O(\sqrt{n})$; the bound $\operatorname{C}(G) \geq \tau_n = \Omega(n^2)$ follows from \cite{pabbaraju2026}.

\subsection{Alternative totalization}\label{app:alternative-totalization}

The descriptions can also be indexed by the block of a satisfied term and its least nonzero position. For fixed $i,p \in [n]$, a description $z$ contains the fields
\[
 (w_1,\ldots,w_n,\sigma_1,\ldots,\sigma_n,j_1,\ldots,j_M),
\]
with the ranges and binary encodings in \Cref{app:description-verification}, so its length is $d' = \Theta(n\ln n)$. The block index $i$ is supplied by the array address. Let $V_{i,p}(x,z)$ be the conjunction of the count, length, and zero checks from that construction, the range checks for the fields listed above, and the conditions $w_{p-1} = 0$ and $w_p = 1$, where $w_0 = 0$. By the proof of \Cref{lem:description-verification}, an accepted description exists if and only if $x$ satisfies a term $(i,S)$ with $p = \min S$. There are $O(n)$ checks, each evaluated deterministically with $O(\ln n)$ queries.

An input to $F'_n$ consists of $x \in \{0,1\}^{n^2}$ and an array $(z_{i,p})_{i,p \in [n]}$ of such descriptions. Define $F'_n$ to be one if and only if $V_{i,p}(x,z_{i,p}) = 1$ for some $i,p \in [n]$. If $f(x) = 0$, every description rejects. Otherwise, the unique satisfied term $(i,S)$ determines the only address $(i,\min S)$ that can accept. This gives a total function on $\Theta(n^3\ln n)$ bits.

To prove the certificate lower bound, consider a partial assignment of fewer than $\tau_n = n(n-k+1)$ bits consistent with the all-zero input to $F'_n$. Group each bit $x_i(p)$ with all bits of $z_{i,p}$. The groups are disjoint, so some block $i$ has at least $k$ groups containing no fixed bit. Choose their positions as $S$ and put $p = \min S$. Set block $i$ to have support $S$ and every other block to zero. This satisfies $(i,S)$, and every bit of $z_{i,p}$ is unfixed, so it can be filled with an accepted description. Set every other unfixed description bit to zero. The resulting positive completion proves $\operatorname{C}(F'_n) \geq \tau_n$. For the upper bound, reading all of $x$ determines whether a term is satisfied and, if so, the only address that can accept. Reading that description then determines $F'_n$. Thus $\operatorname{D}(F'_n) \leq n^2 + d' = O(n^2)$, and $\operatorname{C}(F'_n) = \Theta(n^2)$.

The bounded-error quantum algorithm runs \textsc{FindSource} on $x$, then uses \Cref{lem:prefix-search} to find the least nonzero position $p$ of the returned block $i$. It returns zero if that search returns $\bot$; otherwise it uses Grover search, with error at most $1/100$, to find a rejecting check of $V_{i,p}(x,z_{i,p})$, verifying every reported rejection. It returns one if no rejection is found. On a positive input, the first two steps find the accepting address with probability at least $(5/6)(1-n^{-5})$, and every check there accepts. On a negative input, every description rejects, so any selected description is rejected with probability at least $99/100$. The three steps use $O(\sqrt{n}\ln^2 n)$, $O(\sqrt{n\ln n})$, and $O(\sqrt{n}\ln n)$ queries, respectively, by \Cref{lem:block-finding,lem:prefix-search} and \cite[Theorem~3]{BuhrmanCleveDeWolfZalka}. Hence $\operatorname{Q}(F'_n) = O(\sqrt{n}\ln^2 n)$.

On a positive input to the DNF, \Cref{lem:weighted-success} and the equality $\mu = 0$ on rejection give $1 = \mathbb E\mu \leq 2\Pr[\mathrm{Eliminate}(x) \neq \bot]$. Two independent runs therefore find its term with probability at least $3/4$. The randomized algorithm for $F'_n$ runs these searches and returns zero if no term is returned; otherwise it checks the description at the address determined by a returned term. This gives one-sided error at most $1/4$ and $\operatorname{R}(F'_n) = O(n\ln n)$. Replacing the two runs by \textsc{FindTerm} gives an exact algorithm with the same asymptotic bound. The convention in \Cref{app:implementations} gives its unitary implementation a fixed query count on every input.

For the zero-error lower bound, fix $x_1$ to have support $[k]$ and every other block to zero. Only the description at $(1,1)$ can accept. Fix an accepted description there except for its $n$ count fields, and fix every other description. The remaining $m = n\lceil\log_2(n+1)\rceil$ bits have a unique accepting assignment, since $w_q$ must encode $\min\{q,k\}$. After complementing the coordinates where that assignment is zero, the restriction equals $\operatorname{AND}_m$. By \cite[Proposition~6.1]{Beals2001}, this gives $\operatorname{Q}_0(F'_n) \geq m$. Together with the exact upper bound and $\operatorname{Q}_0(F'_n) \leq \operatorname{Q}_{\mathrm E}(F'_n)$, this proves that both complexities are $\Theta(n\ln n)$.

\end{document}